\documentclass[reqno]{amsart}
\usepackage{amsthm}
\usepackage{amsmath}
\usepackage{amssymb}
\usepackage{url}
\usepackage{enumitem}

\newtheorem{thm}{Theorem}
\newtheorem{prop}[thm]{Proposition}
\newtheorem{lem}[thm]{Lemma}
\newtheorem{cor}[thm]{Corollary}
\theoremstyle{definition}
\newtheorem{defn}[thm]{Definition}
\newtheorem{exam}[thm]{Example}
\theoremstyle{remark}
\newtheorem{rem}[thm]{Remark}

\renewcommand{\leq}{\leqslant}
\renewcommand{\geq}{\geqslant}

\numberwithin{equation}{section}
\numberwithin{thm}{section}

\newcommand{\resetstep}{\setcounter{stepcounter}{0}}
\newcounter{stepcounter}
\newcommand{\step}[1]{\refstepcounter{stepcounter} \medskip {\sl {Step \thestepcounter} --- #1}}

\newcommand{\ZZ}{\mathbb{Z}}
\newcommand{\CC}{\mathbb{C}}
\newcommand{\RR}{\mathbb{R}}
\newcommand{\FF}{\mathbb{F}}
\newcommand{\Acal}{\mathcal A}
\newcommand{\Scal}{\mathcal S}
\newcommand{\Lset}{\mathcal{L}}
\newcommand{\Pset}{\mathcal{P}}
\newcommand{\Ralg}{\mathcal{R}}
\newcommand{\divides}{\mathrel|}
\newcommand{\ndivides}{\mathrel{\nmid}}
\DeclareMathOperator{\ord}{ord}
\DeclareMathOperator{\lcm}{LCM}

\DeclareMathOperator{\MM}{\mathsf{M}}
\DeclareMathOperator{\Cpoly}{\mathsf{C}_{\mathrm{poly}}}
\DeclareMathOperator{\Cpolystar}{\mathsf{C}^\star_{\mathrm{poly}}}
\DeclareMathOperator{\Crefpoly}{\tilde{\mathsf{C}}_{\mathrm{poly}}}
\DeclareMathOperator{\Crefpolystar}{\tilde{\mathsf{C}}^\star_{\mathrm{poly}}}
\DeclareMathOperator{\Cad}{\mathsf{C}_{\mathrm{ad}}}
\DeclareMathOperator{\Cadstar}{\mathsf{C}^\star_{\mathrm{ad}}}
\DeclareMathOperator{\Crefad}{\tilde{\mathsf{C}}_{\mathrm{ad}}}
\DeclareMathOperator{\Crefadstar}{\tilde{\mathsf{C}}^\star_{\mathrm{ad}}}
\DeclareMathOperator{\Cint}{\mathsf{C}_{\mathrm{int}}}
\DeclareMathOperator{\Cintstar}{\mathsf{C}^\star_{\mathrm{int}}}
\DeclareMathOperator{\Tcyc}{\mathsf{T}}

\begin{document}

\title[Faster integer and polynomial multiplication]{Faster integer and polynomial multiplication using cyclotomic coefficient rings}
\author{David Harvey}
\email{d.harvey@unsw.edu.au}
\address{School of Mathematics and Statistics, University of New South Wales, Sydney NSW 2052, Australia}
\author{Joris van der Hoeven}
\email{vdhoeven@lix.polytechnique.fr}
\address{Laboratoire d'informatique, UMR 7161 CNRS, \'Ecole polytechnique, 91128 Palaiseau Cedex, France}

\begin{abstract}
We present an algorithm that computes the product of two $n$-bit integers in $O(n \log n \, (4\sqrt 2)^{\log^* n})$ bit operations.
Previously, the best known bound was $O(n \log n \, 6^{\log^* n})$.
We also prove that for a fixed prime~$p$, polynomials in $\FF_p[X]$ of degree $n$ may be multiplied in $O(n \log n \, 4^{\log^* n})$ bit operations; the previous best bound was $O(n \log n \, 8^{\log^* n})$.
\end{abstract}

\maketitle

\thispagestyle{empty}

\section{Introduction}
\label{sec:intro}

In this paper we present new complexity bounds for multiplying integers and polynomials over finite fields.
Our focus is on theoretical bounds rather than practical algorithms.
We work in the deterministic multitape Turing model~\cite{Pap-complexity}, in which time complexity is defined by counting the number of steps, or equivalently, the number of `bit operations', executed by a Turing machine with a fixed, finite number of tapes.
The main results of the paper also hold in the Boolean circuit model.

The following notation is used throughout.
For $x \in \RR$, we denote by $\log^* x$ the iterated logarithm, that is, the least non-negative integer $k$ such that $\log^{\circ k} x \leq 1$, where $\log^{\circ k} x := \log \cdots \log x$ (iterated $k$ times).
For a positive integer $n$, we define $\lg n := \max(1, \lceil \log_2 n \rceil)$; in particular, expressions like $\lg \lg \lg n$ are defined and take positive values for all $n \geq 1$.
We denote the $n$-th cyclotomic polynomial by $\phi_n(X) \in \ZZ[X]$, and the Euler totient function by $\varphi(n)$.

All absolute constants in this paper are in principle effectively computable.
This includes the implied constants in all uses of $O(\cdot)$ notation.

\subsection{Integer multiplication}

Let $\MM(n)$ denote the number of bit operations required to multiply two $n$-bit integers.
For over 35 years, the best known bound for $\MM(n)$ was that achieved by the Sch\"onhage--Strassen algorithm \cite{SS-multiply}, namely
\begin{equation}
\label{eq:ss-int}
  \MM(n) = O(n \lg n \lg \lg n).
\end{equation}
In 2007, F\"urer described an asymptotically faster algorithm that achieves
\begin{equation}
\label{eq:furer}
  \MM(n) = O(n \lg n \, K_\ZZ^{\log^* n})
\end{equation}
for some unspecified constant $K_\ZZ > 1$ \cite{Fur-faster1,Fur-faster2}.
His algorithm reduces a multiplication of size $n$ to a large collection of multiplications of size exponentially smaller than~$n$; these smaller multiplications are handled recursively.
The $K_\ZZ^{\log^* n}$ term may be understood roughly as follows: the number of recursion levels is $\log^* n + O(1)$, and the constant $K_\ZZ$ measures the amount of `data expansion' that occurs at each level, due to phenomena such as zero-padding.

Immediately following F\"urer's work, De, Kurur, Saha and Saptharishi described a variant based on modular arithmetic \cite{DKSS-intmult}, instead of the approximate complex arithmetic used by F\"urer.
Their algorithm also achieves \eqref{eq:furer}, again for some unspecified $K_\ZZ > 1$.

The first explicit value for $K_\ZZ$ was given by Harvey, van der Hoeven and Lecerf, who described an algorithm that achieves \eqref{eq:furer} with $K_\ZZ = 8$ \cite{HvdHL-mul}.
Their algorithm borrows some important ideas from F\"urer's work, but also differs in several respects.
In particular, their algorithm has no need for the `fast' roots of unity that were the cornerstone of F\"urer's approach (and of the variant of \cite{DKSS-intmult}).
The main presentation in \cite{HvdHL-mul} is based on approximate complex arithmetic, and the paper includes a sketch of a variant based on modular arithmetic that also achieves $K_\ZZ = 8$.

In a recent preprint, the first author announced that in the complex arithmetic case, the constant may be reduced to $K_\ZZ = 6$, by taking advantage of new methods for truncated integer multiplication~\cite{Har-truncmul}.
This improvement does not seem to apply to the modular variants.

The first main result of this paper is the following further improvement.
\begin{thm}
\label{thm:int}
There is an integer multiplication algorithm that achieves
 \[ \MM(n) = O(n \lg n \, (4\sqrt 2)^{\log^* n}). \]
\end{thm}
In other words, \eqref{eq:furer} holds with $K_\ZZ = 4\sqrt 2 \approx 5.657$.
The proof is given in Section~\ref{sec:int}.
The new algorithm works with modular arithmetic throughout; curiously, there is no obvious analogue based on approximate complex arithmetic.

There have been several proposals in the literature for algorithms that achieve $K_\ZZ = 4$ under various unproved number-theoretic conjectures: see \cite[\S9]{HvdHL-mul}, \cite{HvdH-vanilla}, and \cite{CT-zmult}.
Whether $K_\ZZ = 4$ can be reached unconditionally remains an important open question.

\subsection{Polynomial multiplication over finite fields}

For a prime $p$, let $\MM_p(n)$ denote the number of bit operations required to multiply two polynomials in $\FF_p[X]$ of degree less than $n$.
The optimal choice of algorithm for this problem depends very much on the relative size of $n$ and $p$.

If $n$ is not too large compared to $p$, say $\lg n = O(\lg p)$, then a reasonable choice is \emph{Kronecker substitution}: one lifts the polynomials to $\ZZ[X]$, packs the coefficients of each polynomial into a large integer (i.e., evaluates at $X = 2^b$ for $b := 2 \lg p + \lg n$), multiplies these large integers, unpacks the resulting coefficients to obtain the product in $\ZZ[X]$, and finally reduces the output modulo $p$.
This leads to the bound
\begin{equation}
\label{eq:ks-bound}
 \MM_p(n) = O(\MM(n \lg p)) = O(n \lg p \lg(n \lg p) K_\ZZ^{\log^*(n \lg p)}),
\end{equation}
where $K_\ZZ$ is any admissible constant in \eqref{eq:furer}.
To the authors' knowledge, this is the best known asymptotic bound for $\MM_p(n)$ in the region $\lg n = O(\lg p)$.

When $n$ is large compared to $p$, the situation is starkly different.
The Kronecker substitution method leads to poor results, due to coefficient growth in the lifted product: for example, when $p$ is fixed, Kronecker substitution yields
 \[ \MM_p(n) = O(\MM(n \lg n)) = O(n (\lg n)^2 K_\ZZ^{\log^* n}). \]
For many years, the best known bound in this regime was that achieved
by the algebraic version of the Sch\"onhage--Strassen algorithm \cite{SS-multiply,Sch-char2}, namely
\begin{equation}
\label{eq:ss-poly}
 \MM_p(n) = O(n \lg n \lg \lg n \lg p + n \lg n \MM(\lg p)).
\end{equation}
The first term arises from performing $O(n \lg n \lg \lg n)$ additions in $\FF_p$, and the second term from $O(n \lg n)$ multiplications in $\FF_p$.
(In fact, this sort of bound holds for polynomial multiplication over quite general rings \cite{CK-fastmult}.)
For fixed $p$, this is faster than the Kronecker substitution method by a factor of almost $\lg n$.
The main reason for its superiority is that it exploits the modulo $p$ structure throughout the algorithm, whereas the Kronecker substitution method forgets this structure in the very first step.

After the appearance of F\"urer's algorithm, it was natural to ask whether a F\"urer-type bound could be proved for $\MM_p(n)$, in the case that $n$ is large compared to $p$.
This question was answered in the affirmative by Harvey, van der Hoeven and Lecerf, who gave an algorithm that achieves
 \[ \MM_p(n) = O(n \lg p \lg(n \lg p) \, 8^{\log^*(n \lg p)}), \]
\emph{uniformly} for all $n$ and $p$ \cite{HvdHL-ffmul}.
This is a very elegant bound; however, written in this way, it obscures the fact that the constant $8$ plays two quite different roles in the complexity analysis.
One source of the value $8$ is the constant ${K_\ZZ = 8}$ arising from the \emph{integer} multiplication algorithm mentioned above, but there is also a separate constant $K_\FF = 8$ arising from the \emph{polynomial} part of the algorithm.
There is no particular reason to expect that $K_\ZZ = K_\FF$, and it is somewhat of a coincidence that they have the same numerical value in \cite{HvdHL-ffmul}.

To clarify the situation, we mention that one may derive a complexity bound for the algorithm of \cite{HvdHL-ffmul} under the assumption that one has available an integer multiplication algorithm achieving \eqref{eq:furer} for some $K_\ZZ \geq 1$, where possibly $K_\ZZ \neq 8$.
Namely, one finds that
\begin{equation}
\label{eq:ffmul}
 \MM_p(n) = O(n \lg p \lg(n \lg p) \, K_\FF^{\max(0, \log^* n - \log^* p)} K_\ZZ^{\log^* p})
\end{equation}
where $K_\FF = 8$ (we omit the proof).
The second main result of this paper, proved in Section \ref{sec:poly}, is the following improvement in the value of $K_\FF$.
\begin{thm}
\label{thm:poly}
Let $K_\ZZ \geq 1$ be any constant for which \eqref{eq:furer} holds (for example, by Theorem \ref{thm:int}, one may take $K_\ZZ = 4\sqrt 2$).
Then there is a polynomial multiplication algorithm that achieves
\begin{equation}
 \MM_p(n) = O(n \lg p \lg(n \lg p) \, 4^{\max(0, \log^* n - \log^* p)} K_\ZZ^{\log^* p}),
\end{equation}
uniformly for all $n \geq 1$ and all primes~$p$.
\end{thm}
In other words, \eqref{eq:ffmul} holds with $K_\FF = 4$.
In particular, for fixed $p$, one can multiply polynomials in $\FF_p[X]$ of degree $n$ in $O(n \lg n \, 4^{\log^* n})$ bit operations.

Theorem \ref{thm:poly} may be generalised in various ways.
We briefly mention a few possibilities along the lines of \cite[\S8]{HvdHL-ffmul} (no proofs will be given).
First, we may obtain analogous bit complexity bounds for multiplication in $\FF_{p^a}[X]$ and $(\ZZ/p^a \ZZ)[X]$ for $a \geq 1$, and in $(\ZZ/m\ZZ)[X]$ for arbitrary $m \geq 1$ (see Theorems 8.1--8.3 in \cite{HvdHL-ffmul}).
We may also obtain complexity bounds for polynomial multiplication in various algebraic complexity models.
For example, we may construct a straight-line program that multiplies two polynomials in $\Acal[X]$ of degree less than~$n$, for any $\FF_p$-algebra~$\Acal$, using $O(n \lg n \, 4^{\log^* n})$ additions and scalar multiplications and $O(n \, 2^{\log^* n})$ nonscalar multiplications (compare with \cite[Thm.~8.4]{HvdHL-ffmul}).

\subsection{Overview of the new algorithms}
\label{sec:overview}

To explain the new approach, let us first recall the idea behind the polynomial multiplication algorithm of \cite{HvdHL-ffmul}.

Consider a polynomial multiplication problem in $\FF_p[X]$, where the degree $n$ is very large compared to $p$.
By splitting the inputs into chunks, we convert this to a bivariate multiplication problem in $\FF_p[Y,Z]/(f(Y), Z^m - 1)$, for a suitable integer~$m$ and irreducible polynomial $f \in \FF_p[Y]$.
This bivariate product is handled by means of DFTs (discrete Fourier transforms) of length $m$ over $\FF_p[Y]/f$.
The key innovation of \cite{HvdHL-ffmul} was to choose $\deg f$ so that $p^{\deg f} - 1$ is divisible by many small primes, so many, in fact, that their product is comparable to $n$, even though $\deg f$ itself is exponentially smaller.
This is possible thanks to a number-theoretic result of Adleman, Pomerance and Rumely~\cite{APR-primes}, building on earlier work of Prachar~\cite{Pra-divisors}.
Taking $m$ to be a product of many of these primes, we obtain $m \divides p^{\deg f} - 1$, and hence $\FF_p[Y]/f$ contains a root of unity of order $m$.
As $m$ is highly composite, each DFT of length $m$ may be converted to a collection of much smaller DFTs via the Cooley--Tukey method.
These in turn are converted into multiplication problems using Bluestein's algorithm.
These multiplications, corresponding to exponentially smaller values of~$n$, are handled recursively.

The recursion continues until $n$ becomes comparable to $p$.
The number of recursion levels during this phase is $\log^* n - \log^* p + O(1)$, and the constant $K_\FF = 8$ represents the expansion factor at each recursion level.
When $n$ becomes comparable to $p$, the algorithm switches strategy to Kronecker substitution combined with ordinary integer multiplication.
This phase contributes the $K_\ZZ^{\log^* p}$ term.

It was pointed out in \cite[\S8]{HvdHL-ffmul-preprint} that the value of $K_\FF$ can be improved to $K_\FF = 4$ if one is willing to accept certain unproved number-theoretic conjectures, including Artin's conjecture on primitive roots.
More precisely, under these conjectures, one may find an irreducible $f$ of the form $f(Y) = Y^{\alpha-1} + \cdots + Y + 1$, where~$\alpha$ is prime, so that $\FF_p[Y,Z]/(f(Y), Z^m - 1)$ is a direct summand of $\FF_p[Y,Z]/(Y^\alpha - 1, Z^m - 1)$.
This last ring is isomorphic to $\FF_p[X]/(X^{\alpha m} - 1)$, and one may use this isomorphism to save a factor of two in zero-padding at each recursion level.
These savings lead directly to the improved value for $K_\FF$.

To prove Theorem \ref{thm:poly}, we will pursue a variant of this idea.
We will take $f$ to be a cyclotomic polynomial $\phi_\alpha(Y)$ for a judiciously chosen integer $\alpha$ (not necessarily prime).
Since $\phi_\alpha \divides Y^\alpha - 1$, we may use the above isomorphism to realise the same economy in zero-padding as in the conjectural construction of \cite[\S8]{HvdHL-ffmul-preprint}.
However, unlike \cite{HvdHL-ffmul-preprint}, we do not require that $f$ be irreducible in $\FF_p[Y]$.
Thus $\FF_p[Y]/f$ is no longer in general a field, but a direct sum of fields.
The situation is reminiscent of F\"urer's algorithm, in which the coefficient ring $\CC[Y]/(Y^{2^r}+1)$ is not a field, but a direct sum of copies of $\CC$.
The key technical contribution of this paper is to show that we have enough control over the factorisation of $\phi_\alpha$ in $\FF_p[Y]$ to ensure that $\FF_p[Y]/\phi_\alpha$ contains suitable principal roots of unity.
This approach avoids Artin's conjecture and other number-theoretic difficulties, and enables us to reach $K_\FF = 4$ unconditionally.
The construction of $\alpha$ is the subject of Section \ref{sec:cyclotomic}, and the main polynomial multiplication algorithm is presented in Section \ref{sec:poly}.

Let us now outline how we go about proving Theorem \ref{thm:int} (the integer case).
The algorithm is heavily dependent on the polynomial multiplication algorithm just sketched.
We take the basic problem to be multiplication in $\ZZ/(2^n - 1)\ZZ$, for arbitrary positive $n$.
We choose a collection of small primes~$p$, each having around $2 \lg \lg n$ bits, and whose product~$P$ has $(\lg n)^{1+o(1)}$ bits.
By cutting the input integers into many small chunks, we convert to a multiplication in $(\ZZ/P\ZZ)[X]/(X^N - 1)$ for a suitable $N \approx 2n/\lg P$.
One technical headache is that~$n$ is not necessarily divisible by $N$; following \cite{HvdHL-mul}, we deal with this by adapting an idea of Crandall and Fagin~\cite{CF-DWT}.
Next, by the Chinese remainder theorem, we reduce to multiplying in $\FF_p[X]/(X^N - 1)$ for each~$p$ separately.
This is reminisicent of Pollard's algorithm~\cite{Pol-ntt}, but instead of using three primes, here the number of primes grows with~$n$.
At this stage, the coefficient size $\lg p$ is doubly exponentially smaller than~$N$.
We perform these multiplications in $\FF_p[X]/(X^N - 1)$ by applying \emph{two recursion levels} of the polynomial multiplication algorithm of Theorem \ref{thm:poly}.
This reduces the problem to a collection of multiplication problems in $\FF_p[X]$, each doubly exponentially smaller than the original problem.
Using Kronecker substitution, these are converted back to multiplications in $\ZZ/(2^{n'} - 1)\ZZ$, where $n'$ is doubly exponentially smaller than $n$, and the algorithm is applied recursively.

In effect, each recursive call in the new integer multiplication algorithm corresponds to two recursion levels of the existing F\"urer-type algorithms.
The speedup relative to \cite{HvdHL-mul} may be understood as follows.
In the algorithm of \cite{HvdHL-mul}, at each recursion level we incur a factor of two in overhead due to the zero-padding that occurs when we split the inputs into small chunks.
In the new algorithm, the passage from $\ZZ/P\ZZ$ to $\FF_p$ manages the same exponential size reduction without any zero-padding.
This roughly corresponds to saving a factor of two at every second recursion level of the algorithm of \cite{HvdHL-mul}, and explains the factor of $(\sqrt 2)^{\log^* n}$ overall speedup.


\section{Preliminaries}

\subsection{Logarithmically slow functions}
\label{sec:slow}
Let $x_0 \in \RR$, and let $\Phi : (x_0, \infty) \to \RR$ be a smooth increasing function.
We recall from \cite[\S5]{HvdHL-mul} that $\Phi$ is said to be \emph{logarithmically slow} if there exists an integer $\ell \geq 0$ such that
 \[ (\mathord{\log^{\circ \ell}} \circ \Phi \circ \exp^{\circ \ell})(x) = \log x + O(1) \]
as $x \to \infty$.
For example, the functions $\log(5x)$, $5 \log x$, $(\log x)^5$, and $2^{(\log \log x)^5}$ are logarithmically slow, with $\ell = 0, 1, 2, 3$ respectively.

We will always assume that $x_0$ is chosen large enough to ensure that $\Phi(x) \leq x - 1$ for all $x > x_0$.
According to \cite[Lemma 2]{HvdHL-mul}, this is possible for any logarithmically slow function, and it implies that the iterator $\Phi^*(x) := \min\{k \geq 0 : \Phi^{\circ k}(x) \leq x_0\}$ is well-defined on $\RR$.
It is shown in \cite[Lemma 3]{HvdHL-mul} that this iterator satisfies
\begin{equation}
\label{eq:iterator}
 \Phi^*(x) = \log^* x + O(1)
\end{equation}
as $x \to \infty$.
In other words, logarithmically slow functions are more or less indistinguishable from $\log x$, as far as iterators are concerned.

As in \cite{HvdHL-mul} and \cite{HvdHL-ffmul}, we will use logarithmically slow functions to measure \emph{size reduction} in multiplication algorithms.
The typical situation is that we have a function $T(n)$ measuring the (normalised) cost of a certain multiplication algorithm for inputs of size $n$; we reduce the problem to a collection of problems of size $n_i < \Phi^{\circ\kappa}(n)$ for some $\kappa \geq 1$, leading to a bound for $T(n)$ in terms of the various $T(n_i)$.
Applying the reduction recursively, we wish to convert these bounds into an explicit asymptotic estimate for $T(n)$.
This is achieved via the following `master theorem'.
\begin{prop}
\label{prop:master}
Let $K > 1$, $B \geq 0$, and let $\ell \geq 0$ and $\kappa \geq 1$ be integers.
Let $x_0 \geq \exp^{\circ \ell}(1)$, and let $\Phi : (x_0, \infty) \to \RR$ be a logarithmically slow function such that $\Phi(x) \leq x-1$ for all $x > x_0$.
Assume that $x_1 \geq x_0$ is chosen so that $\Phi^{\circ\kappa}(x)$ is defined for all $x > x_1$.
Then there exists a positive constant $C$ (depending on $x_0$, $x_1$, $\Phi$, $K$, $B$, $\ell$ and $\kappa$) with the following property.

Let $\sigma \geq x_1$ and $L > 0$.
Let $\Scal \subseteq \RR$, and let $T : \Scal \to \RR^{\geq}$ be any function satisfying the following recurrence.
First, $T(y) \leq L$ for all $y \in \Scal$, $y \leq \sigma$.
Second, for all $y \in \Scal$, $y > \sigma$, there exist $y_1, \ldots, y_d \in \Scal$ with $y_i \leq \Phi^{\circ \kappa}(y)$, and weights $\gamma_1, \ldots, \gamma_d \geq 0$ with $\sum_i \gamma_i = 1$, such that
 \[ T(y) \leq K \left(1 + \frac{B}{\log^{\circ \ell} y}\right) \sum_{i=1}^d \gamma_i T(y_i) + L. \]
For all $y \in \Scal$, $y > \sigma$, we then have
 \[ T(y) \leq CL(K^{1/\kappa})^{\log^* y - \log^* \sigma}. \]
\end{prop}
\begin{proof}
The special case $\kappa = 1$, $x_1 = x_0$ is exactly \cite[Prop.~8]{HvdHL-mul}.
We indicate briefly how the proof of \cite[Prop.~8]{HvdHL-mul} must be modified to obtain this more general statement.

The first two paragraphs of the proof of \cite[Prop.~8]{HvdHL-mul} may be read verbatim.
In the third paragraph, the inductive statement is changed to
 \[ T(y) \leq E_1 \cdots E_j L (K^{\lceil j/\kappa \rceil} + \cdots + K + 1), \]
where $j := \Phi^*_\sigma(y)$.
The inductive step is modified slightly: for $0 < j < \kappa$ we use the fact that $y_i \leq \sigma$, and for $j \geq \kappa$ the fact that $\Phi^*_\sigma(y_i) \leq \Phi^*_\sigma(\Phi^{\circ\kappa}(y)) = \Phi^*_\sigma(y) - \kappa$.
With these changes, the proof given in \cite{HvdHL-mul} goes through without difficulty.
\end{proof}

\subsection{Discrete Fourier transforms}
\label{sec:dft}

Let $n \geq 1$ and let $\Ralg$ be a commutative ring in which~$n$ is invertible.
A \emph{principal $n$-th root of unity} is an element $\omega \in \Ralg$ such that $\omega^n = 1$ and such that $\sum_{j=0}^{n-1} \omega^{ij} = 0$ for $i = 1, 2, \ldots, n-1$.
If $m$ is a divisor of~$n$, then $\omega^{n/m}$ is easily seen to be a principal $m$-th root of unity.

Fix a principal $n$-th root of unity $\omega$.
The \emph{discrete Fourier transform} (DFT) of the sequence $(a_0, \ldots, a_{n-1}) \in \Ralg^n$ with respect to $\omega$ is the sequence $(\hat a_0, \ldots, \hat a_{n-1}) \in \Ralg^n$ defined by $\hat a_j := \sum_{i=0}^{n-1} \omega^{ij} a_i$.
Equivalently, $\hat a_j = A(\omega^j)$ where $A = \sum_{i=0}^{n-1} a_i X^i \in \Ralg[X]/(X^n - 1)$.

The \emph{inverse DFT} recovers $(a_0, \ldots, a_{n-1})$ from $(\hat a_0, \ldots, \hat a_{n-1})$.
Computationally it corresponds to a DFT with respect to $\omega^{-1}$, followed by a division by $n$, because
 \[ \frac 1n \sum_{j=0}^{n-1} \omega^{-kj} \hat a_j = \frac 1n \sum_{i=0}^{n-1} \sum_{j=0}^{n-1} \omega^{(i-k)j} a_i = a_k, \qquad k = 0, \ldots, n-1. \]

DFTs may be used to implement cyclic convolutions.
Suppose that we wish to compute $C := AB$ where $A, B \in \Ralg[X]/(X^n - 1)$.
We first perform DFTs to compute $A(\omega^i)$ and $B(\omega^i)$ for $i = 0, \ldots, n-1$.
We then compute $C(\omega^i) = A(\omega^i) B(\omega^i)$ for each $i$, and finally perform an inverse DFT to recover $C \in \Ralg[X]/(X^n - 1)$.

This strategy may be generalised to handle a multidimensional cyclic convolution, that is, to compute $C := AB$ for
 \[ A, B \in \Ralg[X_1, \ldots, X_d]/(X_1^{n_1} - 1, \ldots, X_d^{n_d} - 1). \]
For this, we require that each $n_k$ be invertible in $\Ralg$, and that $\Ralg$ contain a principal $n_k$-th root of unity $\omega_k$ for each $k$.
Let $n := n_1 \cdots n_d$.
We first perform multidimensional DFTs to evaluate $A$ and $B$ at the $n$ points $\{(\omega_1^{j_1}, \ldots, \omega_d^{j_d}) : 0 \leq j_k < n_k\}$.
We then multiply pointwise, and finally recover~$C$ via a multidimensional inverse DFT.

Each multidimensional DFT may be reduced to a collection of one-dimensional DFTs as follows.
We first compute $A(X_1, \ldots, X_{d-1}, \omega_d^j) \in \Ralg[X_1, \ldots, X_{d-1}]$ for each $j = 0, \ldots, n_d - 1$; this involves $n/n_d$ DFTs of length~$n_d$.
We then recursively evaluate each of these polynomials at the $n/n_d$ points $(\omega_1^{j_1}, \ldots, \omega_{d-1}^{j_{d-1}})$.
Altogether, this strategy involves computing $n/n_k$ DFTs of length $n_k$ for each $k = 1, \ldots, d$.

Finally, we briefly recall Bluestein's method \cite{Blu-dft} for reducing a (one-dimensional) DFT to a convolution problem (see also \cite[\S2.5]{HvdHL-mul}).
Let $n \geq 1$ be odd and let $\omega \in \Ralg$ be a principal $n$-th root of unity.
Set $\xi := \omega^{(n+1)/2}$, so that $\xi^2 = \omega$ and $\xi^n = 1$.
Then computing the DFT of a given sequence $(a_0, \ldots, a_{n-1}) \in \Ralg^n$ with respect to~$\omega$ reduces to computing the product of the polynomials
 \[ f(Z) := \sum_{i=0}^{n-1} \xi^{i^2} a_i Z^i, \qquad g(Z) := \sum_{i=0}^{n-1} \xi^{-i^2} Z^i \]
in $\Ralg[Z]/(Z^n - 1)$, plus $O(n)$ auxiliary multiplications in $\Ralg$.
Notice that $g(Z)$ is fixed and does not depend on the input sequence.

\subsection{The Crandall--Fagin algorithm}
\label{sec:crandall-fagin}

Consider the problem of computing a `cyclic' integer product of length $n$, that is, a product $uv$ where $u, v \in \ZZ/(2^n - 1)\ZZ$.
If $N$ and $P$ are positive integers such that $N \divides n$ and $\lg P > 2n/N + \lg N$, then we may reduce the given problem to multiplication in $(\ZZ/P\ZZ)[X]/(X^N - 1)$, by cutting up the integers into chunks of $n/N$ bits.
In this section we briefly recall a variant \cite[\S9.2]{HvdHL-mul} of an algorithm due to Crandall and Fagin \cite{CF-DWT} that achieves the same reduction \emph{without} the assumption that $N \divides n$.

Assume that $N \leq n$ and $\lg P > 2 \lceil n/N \rceil + \lg N + 1$, and that we have available some $\theta \in \ZZ/P\ZZ$ with $\theta^N = 2$.
(This~$\theta$ plays the same role as the real $N$-th root of~$2$ in the original Crandall--Fagin algorithm.)
Set $e_i := \lceil ni / N \rceil$ and $c_i := Ne_i - ni$.
Observe that $e_{i+1} - e_i = \lfloor n/N \rfloor$ or $\lceil n/N \rceil$ for each $i$.
Decompose the inputs as $u = \sum_{i=0}^{N-1} 2^{e_i} u_i$ and $v = \sum_{i=0}^{N-1} 2^{e_i} v_i$ where $0 \leq u_i, v_i < 2^{e_{i+1} - e_i}$ (i.e., a decomposition with respect to a `variable base').
Set $U(X) := \sum_{i=0}^{N-1} \theta^{c_i} u_i$ and $V(X) := \sum_{i=0}^{N-1} \theta^{c_i} v_i$, regarded as polynomials in $(\ZZ/P\ZZ)[X]/(X^N - 1)$, and compute the product $W(X) := U(X) V(X)$.
Then one finds (see \cite[\S9.2]{HvdHL-mul}) that the product $uv$ may be recovered by the formula $uv = \sum_{i=0}^{N-1} 2^{e_i} w_i \pmod{2^n - 1}$, where the $w_i$ are integers in $[0, P)$ defined by $W(X) = \sum_{i=0}^{N-1} \theta^{c_i} w_i$.

To summarise, the problem of computing $uv$ reduces to computing a product in $(\ZZ/P\ZZ)[X]/(X^N - 1)$, together with $O(N)$ auxiliary multiplications in $\ZZ/P\ZZ$, and $O(N (\lg n)^2 + N \lg P)$ bit operations to compute the $e_i$ and to handle the final overlap-add phase (again, see \cite[\S9.2]{HvdHL-mul} for details).

\subsection{Data layout}

In this section we discuss several issues relating to the layout of data on the Turing machine tapes.

Integers will always be stored in the standard binary representation.
If $n$ is a positive integer, then elements of $\ZZ/n\ZZ$ will always be stored as residues in the range $0 \leq x < n$, occupying $\lg n$ bits of storage.

If $\Ralg$ is a ring and $f \in \Ralg[X]$ is a polynomial of degree $n \geq 1$, then an element of $\Ralg[X]/f(X)$ will always be represented as a sequence of~$n$ coefficients in the standard monomial order.
This convention is applied recursively, so for rings of the type $(\Ralg[Y]/f(Y))[X]/g(X)$, the coefficient of $X^0$ is stored first, as an element of $\Ralg[Y]/f(Y)$, followed by the coefficient of $X^1$, and so on.

A multidimensional array of size $n_d \times \cdots \times n_1$, whose entries occupy $b$ bits each, will be stored as a linear array of $b n_1 \cdots n_d$ bits.
The entries are ordered lexicographically in the order $(0, \ldots, 0, 0), (0, \ldots, 0, 1), \ldots, (n_d - 1, \ldots, n_1 - 1)$.
In particular, an element of $(\cdots (\Ralg[X_1]/f_1(X_1)) \cdots )[X_d]/f_d(X_d)$ is represented as an $n_d \times \cdots \times n_1$ array of elements of~$\Ralg$.
We will generally prefer the more compact notation $\Ralg[X_1, \ldots, X_d]/(f_1(X_1), \ldots, f_d(X_d))$.

There are many instances where an $n \times m$ array must be transposed so that its entries can be accessed efficiently either `by columns' or `by rows'.
Using the algorithm of \cite[Lemma 18]{BGS-recurrences}, such a transposition may be achieved in $O(b n m \lg \min(n, m))$ bit operations, where $b$ is the bit size of each entry.
(The idea of the algorithm is to split the array in half along the short dimension, and transpose each half recursively.)

One important application is the following result, which estimates the data rearrangement cost associated to the the Agarwal--Cooley method \cite{AC-convolution} for converting between one-dimensional and multidimensional convolution problems (this is closely related to the Good--Thomas DFT algorithm \cite{Goo-fourier, Tho-physics}).
\begin{lem}
\label{lem:permute}
Let $n, m \geq 2$ be relatively prime, and let $\Ralg$ be a ring whose elements are represented using $b$ bits.
There exists an isomorphism
 \[ \Ralg[X]/(X^{nm} - 1) \cong \Ralg[Y,Z]/(Y^n - 1, Z^m - 1) \]
that may be evaluated in either direction in $O(b n m \lg \min(n, m))$ bit operations.
\end{lem}
\begin{proof}
Let $c := m^{-1} \bmod n$, and let
 \[ \beta : \Ralg[X]/(X^{nm} - 1) \to \Ralg[Y,Z]/(Y^n - 1, Z^m - 1) \]
denote the homomorphism that maps $X$ to $Y^c Z$, and acts as the identity on $\Ralg$.
Suppose that we wish to compute $\beta(F)$ for some input polynomial $F = \sum_{k=0}^{nm-1} F_k X^k \in \Ralg[X]/(X^{nm} - 1)$.
Interpreting the list $(F_0, \ldots, F_{nm-1})$ as an $n \times m$ array, the $(i, j)$-th entry corresponds to $F_{im+j}$.
After transposing the array, which costs $O(b n m \lg \min(n, m))$ bit operations, we have an $m \times n$ array, whose $(j,i)$-th entry is $F_{im+j}$.
Now for each $j$, cyclically permute the $j$-th row by $(jc \bmod n)$ slots; altogether this uses only $O(b n m)$ bit operations.
The result is an $m \times n$ array whose $(j, i)$-th entry is $F_{(i - jc \bmod n)m+j}$, which is exactly the coefficient of $Y^{((i-jc)m+j)c} Z^{(i-jc)m+j} = Y^i Z^j$ in $\beta(F)$.
The inverse map $\beta^{-1}$ may be computed by reversing this procedure.
\end{proof}

\begin{cor}
\label{cor:permute}
Let $n_1, \ldots, n_d \geq 2$ be pairwise relatively prime, let $n := n_1 \cdots n_d$, and let $\Ralg$ be a ring whose elements are represented using $b$ bits.
There exists an isomorphism
 \[ \Ralg[X]/(X^n - 1) \cong \Ralg[X_1, \ldots, X_d]/(X_1^{n_1} - 1, \ldots, X_d^{n_d} - 1) \]
that may be evaluated in either direction in $O(b n \lg n)$ bit operations.
\end{cor}
\begin{proof}
Using Lemma \ref{lem:permute}, we may construct a sequence of isomorphisms
\begin{align*}
 \Ralg[X]/(X^{n_1 \cdots n_d} - 1)
   & \cong \Ralg[X_1, W_2]/(X_1^{n_1} - 1, W_2^{n_2 \cdots n_d} - 1) \\
   & \cong \Ralg[X_1, X_2, W_3]/(X_1^{n_1} - 1, X_2^{n_2} - 1, W_3^{n_3 \cdots n_d} - 1) \\
                    & \cdots \\
                    & \cong \Ralg[X_1, \ldots, X_d]/(X_1^{n_1} - 1, \ldots, X_d^{n_d} - 1),
\end{align*}
the $i$-th of which may be computed in $O(b n \lg n_i)$ bit operations.
The overall cost is $O(\sum_i b n \lg n_i) = O(b n \lg n)$ bit operations.
\end{proof}


\section{Cyclotomic coefficient rings}
\label{sec:cyclotomic}

The aim of this section is to construct certain coefficient rings that play a central role in the multiplication algorithms described later.
The basic idea is as follows.
Suppose that we want to multiply two polynomials in $\FF_p[X]$, and that the degree of the product is known to be at most $n$.
If $N$ is an integer with $N > n$, then by appropriate zero-padding, we may embed the problem in $\FF_p[X]/(X^N - 1)$.
Furthermore, if we have some factorisation $N = \alpha m$, where $\alpha$ and $m$ are relatively prime, then there is an isomorphism
 \[ \FF_p[X]/(X^N - 1) \cong \FF_p[Y,Z]/(Y^\alpha - 1, Z^m - 1), \]
and the latter ring is closely related to
 \[ \FF_p[Y,Z]/(\phi_\alpha(Y), Z^m - 1) \cong (\FF_p[Y]/\phi_\alpha)[Z]/(Z^m - 1) \]
(recall that $\phi_\alpha(Y)$ denotes the $\alpha$-th cyclotomic polynomial).
In particular, computing the product in $(\FF_p[Y]/\phi_\alpha)[Z]/(Z^m - 1)$ recovers `most' of the information about the product in $\FF_p[X]/(X^N - 1)$.

In this section we show how to choose $N$, $\alpha$ and $m$ with the following properties:
\begin{enumerate}
\item $N$ is not much larger than $n$, so that not too much space is `wasted' in the initial zero-padding step;
\item $\varphi(\alpha)$ ($= \deg \phi_\alpha$) is not much smaller than $\alpha$, so that we do not lose much information by working modulo $\phi_\alpha(Y)$ instead of modulo $Y^\alpha - 1$ (this missing information must be recovered by other means);
\item the coefficient ring $\FF_p[Y]/\phi_\alpha$ contains a principal $m$-th root of unity, so that we can multiply in $(\FF_p[Y]/\phi_\alpha)[Z]/(Z^m - 1)$ efficiently by means of DFTs over $\FF_p[Y]/\phi_\alpha$;
\item $m$ is a product of many integers that are exponentially smaller than $n$, so that the DFTs of length $m$ may be reduced to many small DFTs; and
\item $\alpha$ is itself exponentially smaller than $n$.
\end{enumerate}
The last two items ensure that the small DFTs can be converted to multiplication problems of degree exponentially smaller than $n$, to allow the recursion to proceed.

\begin{defn}
\label{defn:admissible}
An \emph{admissible tuple} is a sequence $(q_0, q_1, \ldots, q_e)$ of distinct primes ($e \geq 1$) satisfying the following conditions.
First,
\begin{equation}
\label{eq:q-bound}
 (\lg N)^3 < q_i < 2^{(\lg \lg N)^2}, \qquad i = 0, \ldots, e,
\end{equation}
where $N := q_0 \cdots q_e$.
Second, $q_i - 1$ is squarefree for $i = 1, \ldots, e$, and
\begin{equation}
\label{eq:lambda-bound}
 \lambda(q_0, \ldots, q_e) := \lcm(q_1 - 1, \ldots, q_e - 1) < 2^{(\lg \lg N)^2}.
\end{equation}
(Note that $q_0 - 1$ need not be squarefree, and $q_0$ does not participate in \eqref{eq:lambda-bound}.)

An \emph{admissible length} is a positive integer $N$ of the form $N = q_0 \cdots q_e$ where $(q_0, \ldots, q_e)$ is an admissible tuple.
\end{defn}

If $N$ is an admissible length, we treat $(q_0, \ldots, q_e)$ and $\lambda(N) := \lambda(q_0, \ldots, q_e)$ as auxiliary data attached to $N$.
For example, if an algorithm takes~$N$ as input, we implicitly assume that this auxiliary data is also supplied as part of the input.

\newcommand{\thindot}{\hspace{1pt}\mathord{\cdot}\hspace{1pt}}
\begin{exam}
\label{exam:admissible}
For $n = 10^{100000}$, there is a nearby admissible length
\begin{align*}
 N & = 1000000000000000000156121\ldots \text{(99971 digits omitted)} \ldots26353 \\
   & = q_0 q_1 \cdots q_{6035}
\end{align*}
where
\begin{align*}
 q_0 & = 206658761261792645783, \\
 q_1 & = 36658226833235899 = 1 + 2 \thindot 3 \thindot 11 \thindot 17 \thindot 23 \thindot 29 \thindot 37 \thindot 53 \thindot 59 \thindot 67 \thindot 71 \thindot 89, \\
 q_2 & = 36658244723486119 = 1 + 2 \thindot 3 \thindot 17 \thindot 29 \thindot 47 \thindot 59 \thindot 67 \thindot 73 \thindot 83 \thindot 101 \thindot 109, \\
 q_3 & = 36658319675739343 = 1 + 2 \thindot 3 \thindot 7 \thindot 17 \thindot 29 \thindot 31 \thindot 41 \thindot 47 \thindot 53 \thindot 61 \thindot 89 \thindot 103, \\
 q_4 & = 36658428883190467 = 1 + 2 \thindot 3 \thindot 11 \thindot 31 \thindot 43 \thindot 61 \thindot 71 \thindot 73 \thindot 107 \thindot 109 \thindot 113, \\
     & \cdots \\
 q_{6035} & = 37076481100386859 = 1 + 2 \thindot 3 \thindot 13 \thindot 29 \thindot 31 \thindot 59 \thindot 83 \thindot 97 \thindot 101 \thindot 103 \thindot 107
\end{align*}
and
 \[ \lambda(N) = 2 \thindot 3 \thindot 5 \cdots 113 = 31610054640417607788145206291543662493274686990. \]
\end{exam}

\begin{defn}
Let $p$ be a prime.
An admissible length $N$ is called \emph{$p$-admissible} if $N > p^2$ and $p \ndivides N$ (i.e., $p$ is distinct from $q_0, \ldots, q_e$).
\end{defn}
The following result explains how to choose a $p$-admissible length close to any prescribed target.
\begin{prop}
\label{prop:admissible}
There is an absolute constant $z_1 > 0$ with the following property.
Given as input a prime $p$ and an integer $n > \max(z_1, p^2)$, in $2^{O((\lg \lg n)^2)}$ bit operations we may compute a $p$-admissible length $N$ in the interval
\begin{equation}
\label{eq:N-bound}
 n < N < \left(1 + \frac{1}{\lg n}\right) n.
\end{equation}
\end{prop}
The key ingredient in the proof is the following number-theoretic result of Adleman, Pomerance and Rumely.
\begin{lem}[\protect{\cite[Prop.~10]{APR-primes}}]
\label{lem:APR}
There is an absolute constant $C_1 > 0$ with the following property.
For all $x > 10$, there exists a positive squarefree integer $\lambda_0 < x^2$ such that
 \[ \sum_{\substack{\text{\rm $q$ prime} \\ q-1 \divides \lambda_0}} 1 > \exp(C_1 \log x / \log \log x). \]
\end{lem}

\newcommand{\lambdamax}{\lambda_{\text{max}}}

\begin{proof}[Proof of Proposition \ref{prop:admissible}]
Let $\lambdamax := \lceil 2^{\frac29(\lg \lg n)^2} \rceil$, and for $\lambda \geq 1$ define $f(\lambda)$ to be the number of primes $q$ in the interval $(\lg n)^4 < q \leq \lambdamax + 1$ such that $q - 1 \divides \lambda$ and~${q \neq p}$.
We claim that, provided $n$ is large enough, there exists some squarefree $\lambda_0 \in \{1, \ldots, \lambdamax\}$ such that $f(\lambda_0) > \lg n$.
To see this, apply Lemma \ref{lem:APR} with $x := 2^{\frac19(\lg \lg n)^2}$; for large $n$ we then have
 \[ C_1 \log x / \log \log x > 15 (\log_2 x)^{1/2} = 5 \lg \lg n, \]
so Lemma \ref{lem:APR} implies that there exists a positive squarefree integer $\lambda_0 < x^2 \leq \lambdamax$ for which
 \[ \sum_{\substack{\text{$q$ prime} \\ q-1 \divides \lambda_0}} 1 > \exp(5 \lg \lg n) > (\lg n)^5 \]
and hence
 \[ f(\lambda_0) = \sum_{\substack{(\lg n)^4 < q \leq \lambdamax + 1 \\ \text{$q$ prime, $q \neq p$} \\ q-1 \divides \lambda_0}} 1 > (\lg n)^5 - (\lg n)^4 - 1 > \lg n. \]

We may locate one such $\lambda_0$ by means of the following algorithm (adapted from the proof of \cite[Lemma~4.5]{HvdHL-ffmul}).
First use a sieve to enumerate the primes $q$ in the interval $(\lg n)^4 < q \leq \lambdamax + 1$, and to determine which $\lambda = 1, \ldots, \lambdamax$ are squarefree, in $(\lambdamax)^{1+o(1)}$ bit operations.
Now initialise an array of integers $c_\lambda := 0$ for $\lambda = 1, \ldots, \lambdamax$.
For each $q \neq p$, scan through the array, incrementing those~$c_\lambda$ for which $\lambda$ is squarefree and divisible by $q-1$, and stop as soon as one of the~$c_\lambda$ reaches $\lg n$.
We need only allocate $O(\lg \lg n)$ bits per array entry, so each pass through the array costs $O(\lambdamax \lg \lg n)$ bit operations.
The number of passes is $O(\lambdamax)$, so the total cost of finding a suitable $\lambda_0$ is $O(\lambdamax^2 \lg \lg n) = 2^{O((\lg \lg n)^2)}$ bit operations.
Within the same time bound, we may also easily recover a list of primes $q_1, q_2, \ldots, q_{\lg n}$ for which $q_i - 1 \divides \lambda_0$.

Next, compute the partial products $q_1$, $q_1 q_2$, \ldots, $q_1 q_2 \cdots q_{\lg n}$, and determine the smallest integer $e \geq 1$ for which $q_1 \cdots q_e > n / 2^{\frac12(\lg \lg n)^2}$.
Such an $e$ certainly exists, as $q_1 \cdots q_{\lg n} \geq 2^{\lg n} \geq n$.
Since each $q_i$ occupies $O((\lg \lg n)^2)$ bits, this can all be done in $(\lg n)^{O(1)}$ bit operations.
Also, as
 \[ q_e \leq \lambda_0 + 1 \leq 2^{\frac29(\lg \lg n)^2} + 1 < 2^{\frac14(\lg \lg n)^2} \]
and $q_1 \cdots q_{e-1} \leq n/2^{\frac12(\lg \lg n)^2}$, we find that
 \[ 2^{\frac14(\lg \lg n)^2} < \frac{n}{q_1 \cdots q_e} < 2^{\frac12(\lg \lg n)^2} \]
for large $n$.

Let $q_0$ be the least prime that exceeds $n/(q_1 \cdots q_e)$ and that is distinct from~$p$.
According to \cite{BHP-diff}, the interval $[x - x^{0.525}, x]$ contains at least one prime for all sufficiently large $x$; therefore
\begin{align*}
 q_0 & < \frac{n}{q_1 \cdots q_e} + \left(\frac{n}{q_1 \cdots q_e}\right)^{0.6} \\
     & < \left( 1 + (2^{\frac14(\lg \lg n)^2})^{-0.4}\right) \frac{n}{q_1 \cdots q_e} < \left( 1 + \frac{1}{\lg n} \right) \frac{n}{q_1 \cdots q_e}
\end{align*}
for $n$ sufficiently large.
We may find $q_0$ in $2^{O((\lg \lg n)^2)}$ bit operations, by using trial division to test successive integers for primality.

Set $N := q_0 q_1 \cdots q_e$.
Then \eqref{eq:N-bound} holds, and certainly $N > p^2$ and $p \ndivides N$.
Let us check that $(q_0, \ldots, q_e)$ is admissible, provided $n$ is large enough.
For $i = 1, \ldots, e$ we have
 \[ (\lg N)^3 < (\lg n)^4 < q_i \leq \lambda_0 + 1 < 2^{\frac14(\lg \lg n)^2} < 2^{(\lg \lg N)^2}, \]
and also
 \[ (\lg N)^3 < 2^{\frac14(\lg \lg n)^2} < q_0 < \left(1 + \frac{1}{\lg n}\right) 2^{\frac12(\lg \lg n)^2} < 2^{(\lg \lg N)^2}; \]
this establishes \eqref{eq:q-bound}.
Also, as $q_0 > 2^{\frac14(\lg \lg n)^2} > q_i$ for $i = 1, \ldots, e$, we see that $q_0$ is distinct from $q_1, \ldots, q_e$.
Finally, \eqref{eq:lambda-bound} holds because
 \[ \lcm(q_1 - 1, \ldots, q_e - 1) \divides \lambda_0 \leq 2^{\frac29 (\lg \lg n)^2} < 2^{(\lg \lg N)^2}. \]
This also shows that we may compute the auxiliary data $\lambda(q_0, \ldots, q_e)$ in $2^{O((\lg \lg n)^2)}$ bit operations.
\end{proof}

\begin{rem}
Example \ref{exam:admissible} was constructed by enumerating the smallest primes $q_1, q_2, \ldots$ exceeding $(\lg n)^3$ for which $q_i - 1 \divides 2 \cdot 3 \cdot 5 \cdots 113$, halting just before their product reached~$n$, and then choosing $q_0$ to make~$N$ as close to~$n$ as possible.
The proof of Proposition~\ref{prop:admissible} goes a different way: rather than choosing $\lambda$ first, the proof constructs $q_1, \ldots, q_e$ and $\lambda$ simultaneously.
In particular, one cannot guarantee that~$\lambda$ will be a product of an initial segment of primes, as occurred in the example.
Indeed, the proof of \cite[Prop.~10]{APR-primes} (and of its predecessor~\cite{Pra-divisors}) yields very little information at all about the prime factorisation of $\lambda$.
For further discussion, see \cite[Remark 6.2]{APR-primes}.
\end{rem}

\begin{defn}
\label{defn:decomposition}
Let $p$ be a prime and let $N = q_0 \cdots q_e$ be a $p$-admissible length.
A \emph{$p$-admissible divisor} of $N$ is a positive divisor $\alpha$ of $N$, with $q_0 \divides \alpha$, such that the ring $\FF_p[Y]/\phi_\alpha(Y)$ contains a principal $(q_1 \cdots q_e)$-th root of unity, and such that
\begin{equation}
\label{eq:alpha-bound}
  \lg N < \alpha < 2^{(\lg \lg N)^4}
\end{equation}
and
\begin{equation}
\label{eq:phi-alpha-bound}
 \varphi(\alpha) > \left(1 - \frac{1}{\lg N}\right) \alpha.
\end{equation}
\end{defn}

The next result shows how to construct a $p$-admissible divisor for any sufficiently large $p$-admissible length $N$.
The idea behind the construction is as follows.
Let $\ord_n p$ denote the order of $p$ in the multiplicative group of integers modulo $n$.
For any $\alpha \geq 1$, not divisible by $p$, the ring $\FF_p[Y]/\phi_\alpha(Y)$ is a direct sum of fields of order $p^r$, where $r = \ord_\alpha p$ \cite[Lemma~14.50]{vzGG-compalg3}.
Our goal is to ensure that $p^r - 1$ is divisible by $q_1 \cdots q_e$, so that $\FF_p[Y]/\phi_\alpha(Y)$ contains the desired principal root of unity.
One way to force $q_i$ to divide $p^r - 1$ is simply to choose $\alpha$ divisible by~$q_i$, as this implies that $\ord_{q_i} p \divides r$.
The difficulty is that we cannot do this for \emph{all} $q_i$, because then $\alpha$ would become too large, violating \eqref{eq:alpha-bound}.
Fortunately, we can take advantage of the fact that the $q_i - 1$ share a small common multiple $\lambda = \lambda(N)$; this enables us to take $\alpha$ to be a product of a \emph{small} subset of the $q_i$, in such a way that still every one of $q_1, \ldots, q_e$ divides $p^r - 1$.

\begin{prop}
\label{prop:decomposition}
There is an absolute constant $z_2 > 0$ with the following property.
Given as input a prime $p$ and a $p$-admissible length $N > z_2$, we may compute a $p$-admissible divisor $\alpha$ of $N$, together with the cyclotomic polynomial $\phi_\alpha \in \FF_p[Y]$ and a principal $(q_1 \cdots q_e)$-th root of unity in $\FF_p[Y]/\phi_\alpha$, in $2^{O((\lg \lg N)^4)} p^{1+o(1)}$ bit operations.
\end{prop}
\begin{proof}
We are given as input an admissible tuple $(q_0, \ldots, q_e)$ with $N = q_0 \cdots q_e$, and the squarefree integer $\lambda := \lambda(q_0, \ldots, q_e)$.
Let $\Lset$ be the set of primes dividing~$\lambda$.
By \eqref{eq:lambda-bound} we have $|\Lset| \leq \log_2 \lambda < (\lg \lg N)^2$, and we may compute $\Lset$ in $\lambda^{O(1)} = 2^{O((\lg \lg N)^2)}$ bit operations.

We start by computing a table of values of $\ord_{q_i} p$ for $i = 1, \ldots, e$; note that $p \neq q_i$ by hypothesis, so $\ord_{q_i} p$ is well-defined.
We have $q_i - 1 \divides \lambda$ and hence $\ord_{q_i} p \divides \lambda$ for each $i$.
To compute $\ord_{q_i} p$, we first compute $p \bmod q_i$ in $O(\lg q_i \lg p)$ bit operations, and then repeatedly multiply by $p$ modulo $q_i$ until reaching $1$.
Since $\ord_{q_i} p \leq \lambda$, and there are $e = O(\lg N)$ primes $q_i$, the total cost to compute the table is
 \[ O((\lambda \lg^2 q_i + \lg q_i \lg p) \lg N) = (2^{(\lg \lg N)^2} \lg p)^{O(1)} \]
bit operations.

Using the above table, we construct a certain vector $\sigma = (\sigma_1, \ldots, \sigma_e) \in \{0, 1\}^e$ as follows.
Initialise the vector as $\sigma := (0, \ldots, 0)$.
For each $\ell \in \Lset$, search for the smallest $i = 1, \ldots, e$ such that $\ell \divides \ord_{q_i} p$.
If such an $i$ is found, set $\sigma_i := 1$; if no~$i$ is found, ignore this $\ell$.
The cost of computing $\sigma$ is $O(|\Lset| e (\lg \lambda)^2) = (\lg N)^{O(1)}$ bit operations.

Set $\alpha := q_0 \prod_{i:\sigma_i = 1} q_i$.
To establish \eqref{eq:alpha-bound}, note that the number of $i$ for which $\sigma_i = 1$ is at most $|\Lset|$, so \eqref{eq:q-bound} implies that
 \[ \lg N < q_0 \leq \alpha < (2^{(\lg \lg N)^2})^{|\Lset| + 1} \leq (2^{(\lg \lg N)^2})^{(\lg \lg N)^2} = 2^{(\lg \lg N)^4}. \]
For \eqref{eq:phi-alpha-bound}, first observe that
 \[ \frac{\varphi(\alpha)}{\alpha} = \left(1 - \frac{1}{q_0} \right) \prod_{i : \sigma_i = 1} \left(1 - \frac{1}{q_i} \right) > \left(1 - \frac{1}{(\lg N)^3}\right)^{(\lg \lg N)^2}. \]
Since $-\log(1-\varepsilon) < 2\varepsilon$ for any $\varepsilon \in (0, \frac12)$, we obtain
 \[ -\log\frac{\varphi(\alpha)}{\alpha} < \frac{2(\lg \lg N)^2}{(\lg N)^3} < \frac{1}{\lg N} \]
and hence $\varphi(\alpha)/\alpha > \exp(-1/\lg N) > 1 - 1/\lg N$ for sufficiently large $N$.

Now compute the cyclotomic polynomial $\phi_\alpha \in \FF_p[Y]$ (i.e., the reduction modulo~$p$ of $\phi_\alpha(Y) \in \ZZ[Y]$).
This can be done in $(\alpha \lg p)^{O(1)}$ bit operations, using for example \cite[Algorithm 14.48]{vzGG-compalg3}.
We may then determine the factorisation of $\phi_\alpha$ into irreducibles in $\FF_p[Y]$, say $\phi_\alpha = f_1 \cdots f_k$, in $\alpha^{O(1)} p^{1/2+o(1)}$ bit operations \cite[Thm.~1]{Sho-deterministic}.
Since $p \ndivides \alpha$, the $f_j$ are distinct, and each $f_j$ has degree $r := \ord_\alpha p$ \cite[Lemma~14.50]{vzGG-compalg3}.
In other words, $\FF_p[Y]/\phi_\alpha$ is isomorphic to a direct sum of $k$ copies of $\FF_{p^r}$.

We claim that $q_h \divides p^r - 1$ for all $h = 1, \ldots, e$.
For this, it suffices to prove that $\ord_{q_h} p \divides r$ for each $h$.
Since $\lambda$ is squarefree, it suffices in turn to show that every prime $\ell$ dividing $\ord_{q_h} p$ also divides~$r$.
But for every such $\ell$, the procedure for constructing $\sigma$ must have succeeded in finding \emph{some} $i$ for which $\ell \divides \ord_{q_i} p$ (since at least one value of $i$ works, namely $i=h$).
Then $\sigma_i = 1$ for this $i$, so $q_i \divides \alpha$.
This implies that $\ord_{q_i} p \divides \ord_\alpha p = r$, and hence that $\ell \divides r$.

We conclude that $q_1 \cdots q_e \divides p^r - 1$, so each $\FF_p[Y]/f_j$ contains a primitive root of unity of order $q_1 \cdots q_e$.
As the factorisation of $q_1 \cdots q_e$ is known, we may locate one such primitive root in each $\FF_p[Y]/f_j$ in $\alpha^{O(1)} p^{1+o(1)}$ bit operations \cite{Sho-primitive} (see also \cite[Lemma 3.3]{HvdHL-mul}).
Combining these primitive roots via the Chinese remainder theorem, we obtain the desired principal $(q_1 \cdots q_e)$-th root of unity in $\FF_p[Y]/\phi_\alpha$ in another $(\alpha \lg p)^{O(1)}$ bit operations.
\end{proof}

\begin{rem}
The $p^{1+o(1)}$ term in Proposition \ref{prop:decomposition} arises from the best known deterministic complexity bounds for factoring polynomials and finding primitive roots.
If we permit randomised algorithms, then $p^{1+o(1)}$ may be replaced by $(\lg p)^{O(1)}$.
This has no effect on the main results of this paper.
\end{rem}

\begin{exam}
\label{exam:decomposition}
Continuing with Example \ref{exam:admissible}, let us take $p = 3$.
In the notation of the proof of Proposition \ref{prop:decomposition}, we have $\Lset = \{2, 3, 5, \ldots, 113\}$.
For each $\ell \in \Lset$, let us write $q^{(\ell)}$ for the smallest $q_i$ for which $\ord_{q_i} 3$ is divisible by $\ell$.
Then we have
\begin{align*}
 q^{(2)}   & = q_1,   & q^{(3)}   & = q_1,    & q^{(5)}   & = q_5,   & q^{(7)}   & = q_9,    & q^{(11)}  & = q_1, \\
 q^{(13)}  & = q_5,   & q^{(17)}  & = q_1,    & q^{(19)}  & = q_5,   & q^{(23)}  & = q_1,    & q^{(29)}  & = q_1, \\
 q^{(31)}  & = q_3,   & q^{(37)}  & = q_1,    & q^{(41)}  & = q_3,   & q^{(43)}  & = q_4,    & q^{(47)}  & = q_2, \\
 q^{(53)}  & = q_1,   & q^{(59)}  & = q_1,    & q^{(61)}  & = q_3,   & q^{(67)}  & = q_1,    & q^{(71)}  & = q_1, \\
 q^{(73)}  & = q_2,   & q^{(79)}  & = q_6,    & q^{(83)}  & = q_2,   & q^{(89)}  & = q_1,    & q^{(97)}  & = q_5, \\
 q^{(101)} & = q_2,   & q^{(103)} & = q_3,    & q^{(107)} & = q_4,   & q^{(109)} & = q_2,    & q^{(113)} & = q_4.
\end{align*}
Therefore $\sigma_i = 1$ for $i = 1, 2, 3, 4, 5, 6, 9$, and we have
\begin{align*}
          \alpha & = q_0 q_1 q_2 q_3 q_4 q_5 q_6 q_9 \\
                 & \approx 1.8385309928916569681 \times 10^{136}, \\
 \varphi(\alpha) & \approx 1.8385309928916566171 \times 10^{136}, \\
               r & = \ord_\alpha 3 = 2 \cdot 3 \cdot 5 \cdots 109 \cdot 113 \cdot 883 \cdot 9041 \cdot 327251 \cdot 39551747.
\end{align*}
The ring $\FF_3[Y]/\phi_\alpha$ is isomorphic to a direct sum of $\varphi(\alpha) / r$ copies of $\FF_{3^r}$.
The extraneous factors in $r$ (namely $883$, $9041$, $327251$ and $39551747$) arise from the auxiliary prime $q_0$.
Let
 \[ m := N / \alpha = q_7 q_8 q_{10} q_{11} \cdots q_{6035} \approx 5.439125 \times 10^{99863}; \]
then since $m \divides q_1 \cdots q_e \divides 3^r - 1$, each copy of $\FF_{3^r}$ contains a primitive $m$-th root of unity, so $\FF_3[Y]/\phi_\alpha$ contains a principal $m$-th root of unity.
Thus it is possible to multiply in the ring $\FF_3[Y,Z]/(\phi_\alpha(Y), Z^m - 1)$ by using DFTs over $\FF_3[Y]/\phi_\alpha$.
\end{exam}

\begin{rem}
In Example \ref{exam:decomposition}, every $\ell \in \Lset$ divides $\ord_{q_i} p$ for some~$i$.
It seems likely that this always occurs (at least for large $n$), but we do not know how to prove this.
If it fails for some $\ell$, then $r$ may turn out not to be divisible by $\ell$, but the proof of Proposition \ref{prop:decomposition} shows that we still have $q_h \divides p^r - 1$ for every $h = 1, \ldots, e$.
\end{rem}


\section{Faster polynomial multiplication}
\label{sec:poly}

The goal of this section is to prove Theorem \ref{thm:poly}.
We assume that $K_\ZZ \geq 1$ is a constant for which we have available an integer multiplication algorithm achieving $\MM(n) = O(n \lg n \, K_\ZZ^{\log^* n})$.

We will describe a recursive routine \textsc{PolynomialMultiply}, that takes as input integers $r, t \geq 1$, a prime $p$, and polynomials $U_1, \ldots, U_t, V \in \FF_p[X]/(X^r - 1)$, and computes the products $U_1 V, \ldots, U_t V$.
Its running time is denoted by $\Cpoly(t, r, p)$.
Note that the input polynomials $U_1, \ldots, U_t, V$ are expected to be supplied consecutively on the input tape (first $U_1$, then $U_2$, and so on), and the outputs $U_1 V, \ldots, U_t V$ should also be written consecutively to the output tape.

The role of the parameter $t$ is to allow us to amortise the cost of transforming the fixed operand $V$ across $t$ products.
This optimisation (borrowed from \cite{HvdHL-mul} and \cite{HvdHL-ffmul}) saves a constant factor in time at each recursion level of the main algorithm.
Altogether the algorithm will perform $2t+1$ transforms: $t+1$ forward transforms for $U_1, \ldots, U_t$ and~$V$, followed by $t$ inverse transforms to recover the products $U_1 V, \ldots, U_t V$.

To simplify the analysis, it is convenient to introduce the normalisation
 \[ \Cpolystar(r, p) := \sup_{t \geq 1} \frac{\Cpoly(t,r,p)}{(2t+1) r \lg p \lg(r \lg p)}. \]
We certainly have $\MM_p(n) < \Cpoly(1, 2n, p) + O(n \lg p)$, so to prove Theorem \ref{thm:poly} it is enough to show that
\begin{equation}
\label{eq:Cpolystar-goal}
 \Cpolystar(r, p) = O(4^{\max(0, \log^* r - \log^* p)} K_\ZZ^{\log^* p}).
\end{equation}

The algorithms presented in this section perform many auxiliary multiplications and divisions involving `small' integers and polynomials.
We assume that all auxiliary divisions are reduced to multiplication via Newton's method \cite[Ch.~9]{vzGG-compalg3}, so that the cost of a division (by a monic divisor) is at most a constant multiple of the cost of a multiplication of the same bit size.
We also assume that, unless otherwise specified, all auxiliary multiplications are handled using the integer and polynomial variants of the Sch\"onhage--Strassen algorithm, whose complexities are given by \eqref{eq:ss-int} and \eqref{eq:ss-poly}.

We first discuss a subroutine \textsc{Transform} that handles DFTs over rings of the form $\Ralg_{p,\alpha} := \FF_p[Y]/\phi_\alpha(Y)$, where~$p$ is a prime and $\alpha \geq 1$.
It takes as input $p$ and~$\alpha$, positive integers $t$ and $n$ such that $n$ is odd and relatively prime to $\alpha$, a principal $n$-th root of unity $\omega \in \Ralg_{p,\alpha}$, and $t$ input sequences $(a_{s,0}, \ldots, a_{s,n-1}) \in \Ralg_{p,\alpha}^n$ for $s = 1, \ldots, t$.
Its output is the sequence of transforms $(\hat a_{s,0}, \ldots, \hat a_{s,n-1}) \in \Ralg_{p,\alpha}^n$ with respect to $\omega$, for $s = 1, \ldots, t$.
Just like \textsc{PolynomialMultiply}, the input and output sequences are stored consecutively on the tape.

Let $\Tcyc(t, n, \alpha, p)$ denote the running time of \textsc{Transform}.
The following result shows how to reduce the DFT problem to an instance of \textsc{PolynomialMultiply}.
\begin{prop}
\label{prop:bluestein}
We have
 \[ \Tcyc(t, n, \alpha, p) < \Cpoly(t, n \alpha, p) + O(t n \alpha \lg \alpha \lg \lg \alpha \lg p \lg \lg p \lg \lg \lg p). \]
\end{prop}
\begin{proof}
Let $\Ralg := \Ralg_{p,\alpha}$.
We use Bluestein's method to reduce each DFT to the problem of computing a certain product $f_s(Z) g(Z)$ in $\Ralg[Z]/(Z^n - 1)$, plus $O(n)$ multiplications in $\Ralg$, where $f_s(Z)$ and $g(Z)$ are defined as in Section \ref{sec:dft}.
By \eqref{eq:ss-int} and \eqref{eq:ss-poly}, each multiplication in $\Ralg$ costs
\begin{equation}
\label{eq:Ralg-mult}
 O((\alpha \lg \alpha \lg \lg \alpha) (\lg p \lg \lg p \lg \lg \lg p))
\end{equation}
bit operations. 
To handle the products $f_s(Z) g(Z)$, we first lift the polynomials from $\FF_p[Y,Z]/(\phi_\alpha(Y), Z^n - 1)$ to $\FF_p[Y,Z]/(Y^\alpha - 1, Z^n - 1)$ (for example, by zero-padding in $Y$ up to degree $\alpha$).
We then compute their images under the isomorphism
 \[ \FF_p[Y,Z]/(Y^\alpha - 1, Z^n - 1) \cong \FF_p[X]/(X^{n \alpha} - 1) \]
provided by Lemma \ref{lem:permute}; this costs altogether $O(t n \alpha \lg \alpha \lg p)$ bit operations.
We call \textsc{PolynomialMultiply} to compute the products in $\FF_p[X]/(X^{n \alpha} - 1)$, at a cost of $\Cpoly(t, n\alpha, p)$ bit operations.
We evaluate the inverse of the above isomorphism to bring the products back to $\FF_p[Y,Z]/(Y^\alpha - 1, Z^n - 1)$.
Finally, we reduce modulo $\phi_\alpha(Y)$ to obtain the desired products in $\Ralg[Z]/(Z^n - 1)$; the cost of each of these divisions is given by \eqref{eq:Ralg-mult}.
\end{proof}

We now return to multiplication in $\FF_p[X]/(X^r - 1)$.
Our implementation of \textsc{PolynomialMultiply} chooses one of two algorithms, depending on the size of $r$ relative to $p$.
For $r \leq p^2$ it uses the straightforward Kronecker substitution method described in Section \ref{sec:intro}.
By \eqref{eq:ks-bound} this yields the bound
 \[ \Cpoly(t,r,p) = O(t \MM(r \lg p)) = O(t r \lg p \lg(r \lg p) K_\ZZ^{\log^*(r \lg p)}) \]
and hence
\begin{equation}
\label{eq:Cpolystar-small-r}
 \Cpolystar(r,p) = O(K_\ZZ^{\log^*(p^2 \lg p)}) = O(K_\ZZ^{\log^* p}), \qquad r \leq p^2.
\end{equation}
Therefore \eqref{eq:Cpolystar-goal} holds in this case.

For $r > p^2$, most of the work will be delegated to a subroutine \textsc{AdmissibleMultiply}, which is defined as follows.
It takes as input an integer $t \geq 1$, a prime~$p$, a $p$-admissible length $N$, and polynomials $U_1, \ldots, U_t, V \in \FF_p[X]/(X^N - 1)$, and computes the products $U_1 V, \ldots, U_t V$.
In other words, it has the same interface as as \textsc{PolynomialMultiply}, but it only works for $p$-admissible lengths.
We denote its running time by $\Cad(t, N, p)$.
As above we also define the normalisation
 \[ \Cadstar(N, p) := \sup_{t \geq 1} \frac{\Cad(t,N,p)}{(2t+1) N \lg p \lg(N \lg p)}. \]
The reduction from \textsc{PolynomialMultiply} to \textsc{AdmissibleMultiply} in the case $r > p^2$ is given by the following proposition.

\begin{prop}
\label{prop:Cpoly-bound}
There is an absolute constant $z_3 > 0$ with the following property.
For any prime $p$ and any integer $r > \max(z_3, p^2)$, there exists a $p$-admissible length~$N$ in the interval
\begin{equation}
\label{eq:N-bound-2}
  2r < N < \left(1 + \frac{1}{\lg r}\right) 2r
\end{equation}
such that
\begin{equation}
\label{eq:Cpolystar-bound}
 \Cpolystar(r, p) < \left(2 + \frac{O(1)}{\lg r}\right) \Cadstar(N,p) + O(1).
\end{equation}
\end{prop}
\begin{proof}
Given as input $U_1, \ldots, U_t, V \in \FF_p[X]/(X^r - 1)$, our goal is to compute the products $U_1 V, \ldots, U_t V$.
For sufficiently large $r$ we may apply Proposition \ref{prop:admissible} with $n := 2r$ to find a $p$-admissible length $N$ such that \eqref{eq:N-bound-2} holds.
Since $N > 2r$, we may simply zero-pad to reduce each problem to multiplication in $\FF_p[X]/(X^N - 1)$.
This yields
 \[ \Cpoly(t, r, p) < \Cad(t, N, p) + O(tr \lg p) + 2^{O((\lg \lg r)^2)}, \]
where the $t r \lg p$ term arises from the reduction modulo $X^r - 1$, and the last term from Proposition \ref{prop:admissible}.
Dividing by $(2t+1) r \lg p \lg(r \lg p)$ and taking suprema over $t \geq 1$, we find that
 \[ \Cpolystar(r, p) < \frac{N \lg (N \lg p)}{r \lg (r \lg p)} \Cadstar(N, p) + O(1). \]
Finally, since $\lg(N \lg p) \leq \lg(r \lg p) + 2$ we obtain
 \[ \frac{N \lg (N \lg p)}{r \lg (r \lg p)} < 2 \left( 1 + \frac{1}{\lg r} \right) \left(1 + \frac{2}{\lg(r \lg p)}\right) < 2 + \frac{O(1)}{\lg r}. \qedhere \]
\end{proof}

The motivation for defining admissible lengths is the following result, which shows how to implement \textsc{AdmissibleMultiply} in terms of a large collection of exponentially smaller instances of \textsc{PolynomialMultiply}.

\begin{prop}
\label{prop:Cad-bound}
There is an absolute constant $z_4 > 0$ with the following property.
Let $p$ be a prime and let $N > z_4$ be a $p$-admissible length.
Then there exist integers $r_1, \ldots, r_d$ in the interval
\begin{equation}
\label{eq:ri-bound}
 2^{(\lg \lg N)^6} < r_i < 2^{(\lg \lg N)^7},
\end{equation}
and weights $\gamma_1, \ldots, \gamma_d > 0$ with $\sum_i \gamma_i = 1$, such that
\begin{equation}
\label{eq:Cad-bound}
 \Cadstar(N, p) < \left(2 + \frac{O(1)}{\lg \lg N}\right) \sum_{i=1}^d \gamma_i \Cpolystar(r_i, p) + O(1).
\end{equation}
\end{prop}
\begin{proof}

We are given as input a prime $p$, a $p$-admissible length $N = q_0 \cdots q_e$ and polynomials $U_1, \ldots, U_t, V \in \FF_p[X]/(X^N - 1)$.
Our goal is to compute the products $U_1 V, \ldots, U_t V$.
We will describe a series of reductions that converts this problem to a collection of exponentially smaller multiplication problems, plus overhead of $O(t N \lg N \lg p)$ bit operations incurred during the reductions.

\step{reduce to products over cyclotomic coefficient ring.}
Invoking Proposition \ref{prop:decomposition}, we compute a $p$-admissible divisor $\alpha$ of $N$, the cyclotomic polynomial $\phi_\alpha \in \FF_p[Y]$, and a principal $(q_1 \cdots q_e)$-th root of unity $\omega \in \FF_p[Y]/\phi_\alpha$.
As $p^2 < N$, this requires at most $2^{O((\lg \lg N)^4)} p^{1+o(1)} < N^{1/2+o(1)}$ bit operations.

Set $\psi_\alpha := (Y^\alpha - 1)/\phi_\alpha \in \FF_p[Y]$.
Since $Y^\alpha - 1$ has no repeated factors in $\FF_p[Y]$, we have $(\phi_\alpha, \psi_\alpha) = 1$.
Using the Euclidean algorithm, compute polynomials $\chi_1, \chi_2 \in \FF_p[Y]$ of degree at most $\alpha$ such that $\chi_1 \phi_\alpha + \chi_2 \psi_\alpha = 1$; this costs at most $(\alpha \lg p)^{O(1)} < N^{o(1)}$ bit operations.

Let $m := N/\alpha$.
As $m$ and $\alpha$ are coprime, Lemma \ref{lem:permute} provides an isomorphism
 \[ \FF_p[X]/(X^N - 1) \cong \FF_p[Y,Z]/(Y^\alpha - 1, Z^m - 1) \]
that may be evaluated in either direction in $O(m \alpha \lg \alpha \lg p)$ bit operations.
By~\eqref{eq:alpha-bound} this simplifies to $O(N (\lg \lg N)^4 \lg p) = O(N \lg N \lg p)$ bit operations.
Next, since $(\phi_\alpha, \psi_\alpha) = 1$, there is an isomorphism
 \[ \FF_p[Y]/(Y^\alpha - 1) \cong (\FF_p[Y]/\phi_\alpha) \oplus (\FF_p[Y]/\psi_\alpha). \]
Using the precomputed polynomials $\chi_1$ and $\chi_2$, we may evaluate the above isomorphism in either direction in
\begin{align*}
 O((\alpha \lg \alpha \lg \lg \alpha) (\lg p \lg \lg p \lg \lg \lg p))
   & = O(\alpha (\lg \lg N)^5 (\lg \lg \lg N)^2 \lg p) \\
   & = O(\alpha \lg N \lg p)
\end{align*}
bit operations (here we have again used \eqref{eq:alpha-bound} and the fact that $p^2 < N$).
This isomorphism induces another isomorphism
 \[ \FF_p[Y]/(Y^\alpha - 1, Z^m - 1) \cong (\FF_p[Y]/\phi_\alpha)[Z]/(Z^m - 1) \oplus (\FF_p[Y]/\psi_\alpha)[Z]/(Z^m - 1) \]
by acting on the coefficient of each $Z^i$ separately; it may be evaluated in either direction in $O(m\alpha \lg N \lg p) = O(N \lg N \lg p)$ bit operations.
Chaining these isomorphisms together, we obtain an isomorphism
 \[ \FF_p[X]/(X^N - 1) \cong (\FF_p[Y]/\phi_\alpha)[Z]/(Z^m - 1) \oplus (\FF_p[Y]/\psi_\alpha)[Z]/(Z^m - 1) \]
that may be evaluated in either direction in $O(N \lg N \lg p)$ bit operations.

We now use the following algorithm.
First, at a cost of $O(t N \lg N \lg p)$ bit operations, apply the above isomorphism to $U_1, \ldots, U_t$ and $V$ to obtain polynomials
\begin{align*}
 U'_1, \ldots, U'_t, V' & \in (\FF_p[Y]/\phi_\alpha)[Z]/(Z^m - 1), \\
 \tilde U'_1, \ldots, \tilde U'_t, \tilde V' & \in (\FF_p[Y]/\psi_\alpha)[Z]/(Z^m - 1).
\end{align*}
Second, compute the products $\tilde U'_1 \tilde V', \ldots, \tilde U'_t \tilde V'$: since $\deg \psi_\alpha < \alpha/\lg N$ by \eqref{eq:phi-alpha-bound}, each of these products may be converted, via Kronecker substitution, to a product of univariate polynomials in $\FF_p[X]$ of degree $O(m \alpha / \lg N) = O(N / \lg N)$ (i.e., map $Y$ to $X$ and $Z$ to $X^{2 \deg \psi_\alpha}$).
The cost of these multiplications is
 \[ O(t ((N / \lg N) \lg N \lg \lg N) (\lg p \lg \lg p \lg \lg \lg p)) = O(t N \lg N \lg p) \]
bit operations.
Third, compute the products $U'_1 V', \ldots, U'_t V'$, using the method explained in Step 2 below.
Finally, at a cost of $O(t N \lg N \lg p)$ bit operations, apply the inverse isomorphism to the pairs $(U'_s V', \tilde U'_s \tilde V')$ to obtain the desired products $U_1 V, \ldots, U_t V$.

\step{convert to multidimensional convolutions.}
Let $\Ralg := \FF_p[Y]/\phi_\alpha$.
In this step our goal is to compute the products $U'_1 V', \ldots, U'_t V'$, where $U'_1, \ldots, U'_t, V' \in \Ralg[Z]/(Z^m - 1)$.
We do this by converting each problem to a multidimensional convolution of size $m_d \times \cdots \times m_1$, for a suitable decomposition $m = m_1 \cdots m_d$.
For the subsequent complexity analysis, it is important that the $m_i$ are chosen to be somewhat larger than the coefficient size.
To achieve this we proceed as follows.

Let $m = \ell_1 \cdots \ell_u$ be the prime factorisation of $m$.
The $\ell_j$ form a subset of $\{q_1, \ldots, q_e\}$, so by \eqref{eq:q-bound} we have
\begin{equation}
\label{eq:ellj-bound}
 (\lg N)^3 < \ell_j < 2^{(\lg \lg N)^2}
\end{equation}
for each $j$.
Let $w := \lfloor \frac25 (\lg \lg N)^5 \rfloor$.
We certainly have $u > w$ for large enough $N$, as \eqref{eq:ellj-bound} and \eqref{eq:alpha-bound} imply that
 \[ u > \frac{\log_2 m}{(\lg \lg N)^2}
      = \frac{\log_2 N - \log_2 \alpha}{(\lg \lg N)^2}
      > \frac{\log_2 N - (\lg \lg N)^4}{(\lg \lg N)^2} \gg (\lg \lg N)^5. \]
Therefore we may take
\begin{align*}
 m_1     & := \ell_1 \cdots \ell_w, \\
 m_2     & := \ell_{w+1} \cdots \ell_{2w}, \\
         & \cdots \\
 m_{d-1} & := \ell_{(d-2)w+1} \cdots \ell_{(d-1)w}, \\
 m_d     & := \ell_{(d-1)w+1} \cdots \ell_{dw} \ell_{dw+1} \cdots \ell_u,
\end{align*}
where $d := \lfloor u/w \rfloor \geq 1$.
Each $m_i$ is a product of exactly $w$ primes, except possibly~$m_d$, which is a product of at least~$w$ and at most $2w-1$ primes.
For large $N$ we have
\begin{equation}
\label{eq:mi-upper-bound}
 m_i < (2^{(\lg \lg N)^2})^{2w} \leq 2^{\frac45(\lg \lg N)^7}
\end{equation}
and
\begin{equation}
\label{eq:mi-lower-bound}
 m_i > ((\lg N)^3)^w \geq (2^{\lg \lg N - 1})^{3w} > 2^{(\lg \lg N)^6}
\end{equation}
for all $i$, and hence
\begin{equation}
\label{eq:d-bound}
 d \leq \frac{\log_2 m}{(\lg \lg N)^6} \leq \frac{\lg N}{(\lg \lg N)^6}.
\end{equation}
Computing the decomposition $m = m_1 \cdots m_d$ requires no more than $(\lg N)^{O(1)}$ bit operations.

As the $m_i$ are pairwise relatively prime, Corollary \ref{cor:permute} furnishes an isomorphism
 \[ \Ralg[Z]/(Z^m - 1) \cong \Ralg[Z_1, \ldots, Z_d]/(Z_1^{m_1} - 1, \ldots, Z_d^{m_d} - 1) \]
that may be computed in either direction in $O((m \lg m) (\alpha \lg p)) = O(N \lg N \lg p)$ bit operations.
Therefore we may use the following algorithm.
First, at a cost of $O(t N \lg N \lg p)$ bit operations, compute the images
 \[ U''_1, \ldots, U''_t, V'' \in \Ralg[Z_1, \ldots, Z_d]/(Z_1^{m_1} - 1, \ldots, Z_d^{m_d} - 1) \]
of $U'_1, \ldots, U'_t, V'$ under the above isomorphism.
Next, as explained in Step 3 below, compute the products $U''_1 V'', \ldots, U''_t V''$.
Finally, apply the inverse isomorphism to recover the products $U'_1 V', \ldots, U'_t V'$; again this costs $O(t N \lg N \lg p)$ bit operations.

\step{reduce to DFTs over $\Ralg$.}
In this step our goal is to compute the products $U''_1 V'', \ldots, U''_t V''$, where $U''_1, \ldots, U''_t$ and $V''$ are as above.
Let $\omega_i := \omega^{q_1 \cdots q_e/m_i}$ for $i = 1, \ldots, d$, where $\omega$ is the principal $(q_1 \cdots q_e)$-th root of unity in $\Ralg$ computed in Step~1.
According to the discussion in Section \ref{sec:dft}, the desired multidimensional convolutions may be computed by performing $t+1$ multidimensional $m$-point DFTs with respect to the evaluation points $(\omega_1^{j_1}, \ldots, \omega_d^{j_d})$, followed by $tm$ pointwise multiplications in $\Ralg$, and then $t$ multidimensional $m$-point inverse DFTs and $tm$ divisions by $m$.
The total cost of the pointwise multiplications and divisions is
 \[ O(t m (\alpha \lg \alpha \lg \lg \alpha) (\lg p \lg \lg p \lg \lg \lg p)) = O(t N \lg N \lg p) \]
bit operations.

Each of the $2t+1$ multidimensional DFTs may be converted to a collection of one-dimensional DFTs of lengths $m_1, \ldots, m_d$ by the method explained in Section~\ref{sec:dft}.
Note that the inputs must be rearranged so that the data to transform along each dimension may be accessed sequentially.
Let $1 \leq i \leq d$, and consider the transforms of length $m_i$.
Treating each input vector as a sequence of $m_{i+1} \cdots m_d$ arrays of size $m_i \times (m_1 \cdots m_{i-1})$, we must transpose each array into an array of size $(m_1 \cdots m_{i-1}) \times m_i$, perform $m/m_i$ DFTs of length $m_i$, and then transpose back to the original ordering.
The total cost of all these transpositions is
 \[ O(t m \alpha \lg p \textstyle \sum_i \lg m_i) = O(t N \lg p \lg m) = O(t N \lg N \lg p) \]
bit operations.

The one-dimensional DFTs over $\Ralg$ are handled by the \textsc{Transform} subroutine.
Combining the contributions from Steps 1, 2 and 3 shows that
 \[ \Cad(t, N, p) < (2t + 1) \sum_{i=1}^d \Tcyc\Big(\frac{m}{m_i}, m_i, \alpha, p\Big) + O(t N \lg N \lg p). \]

This concludes the description of the algorithm; it remains to establish the overall complexity claim.
First, Proposition \ref{prop:bluestein} yields
\begin{multline*}
 \sum_{i=1}^d \Tcyc\Big(\frac{m}{m_i}, m_i, \alpha, p\Big) < \sum_{i=1}^d \Cpoly\Big(\frac{m}{m_i}, m_i \alpha, p\Big) \\
  + O(d m \alpha \lg \alpha \lg \lg \alpha \lg p \lg \lg p \lg \lg \lg p).
\end{multline*}
By \eqref{eq:d-bound}, the last term lies in
 \[ O(d N (\lg \lg N)^5 (\lg \lg \lg N)^2 \lg p) = O(N \lg N \lg p). \]
Setting $r_i := m_i \alpha$ for $i = 1, \ldots, d$, we obtain
 \[ \Cad(t, N, p) < (2t + 1) \sum_{i=1}^d \Cpoly\Big(\frac{N}{r_i}, r_i, p\Big) + O(t N \lg N \lg p). \]
Notice that \eqref{eq:ri-bound} follows immediately from \eqref{eq:mi-upper-bound}, \eqref{eq:mi-lower-bound} and \eqref{eq:alpha-bound} (for large $N$).
For the normalised quantities, we have
\begin{align*}
 \Cadstar(N, p) & < \sum_{i=1}^d \frac{\Cpoly\big(\frac{N}{r_i}, r_i, p\big)}{N \lg p \lg (N \lg p)} + O(1) \\
                & < \sum_{i=1}^d \Big(\frac{2N}{r_i} + 1\Big) \frac{r_i \lg (r_i \lg p)}{N \lg(N \lg p)} \Cpolystar(r_i, p) + O(1).
\end{align*}
Now observe that
 \[ \frac{\lg(r_i \lg p)}{\lg(N \lg p)} < \frac{\log_2 m_i + \log_2 \alpha + \lg \lg p + O(1)}{\log_2 N} < \frac{\log_2 m_i + O((\lg \lg N)^4)}{\log_2 m}. \]
Put $\gamma_i := \log_2 m_i / \log_2 m$, so that $\sum_i \gamma_i = 1$.
Then \eqref{eq:mi-lower-bound} implies that
 \[ \frac{\lg(r_i \lg p)}{\lg(N \lg p)} < \left(1 + \frac{O((\lg \lg N)^4)}{\log_2 m_i}\right) \gamma_i < \left(1 + \frac{O(1)}{\lg \lg N}\right) \gamma_i. \]
Moreover, from \eqref{eq:ri-bound} we certainly have
 \[ \Big(\frac{2N}{r_i} + 1\Big) \frac{r_i}{N} = 2 + \frac{r_i}{N} < 2 + \frac{O(1)}{\lg \lg N}. \]
The desired bound \eqref{eq:Cad-bound} follows immediately.
\end{proof}

Combining Proposition \ref{prop:Cpoly-bound} and Proposition \ref{prop:Cad-bound}, we obtain the following recurrence inequality for $\Cpolystar(r, p)$.
(This is identical to Theorem 7.1 of \cite{HvdHL-ffmul}, but with the constant $8$ replaced by $4$.)
\begin{prop}
\label{prop:Cpoly-recurrence}
There are absolute constants $z_5, C_2, C_3 > 0$ and a logarithmically slow function $\Phi : (z_5, \infty) \to \RR$ with the following property.
For any prime $p$ and any integer $r > \max(z_5, p^2)$, there exist positive integers $r_1, \ldots, r_d < \Phi(r)$, and weights $\gamma_1, \ldots, \gamma_d > 0$ with $\sum_i \gamma_i = 1$, such that
\begin{equation}
\label{eq:Cpoly-bound}
 \Cpolystar(r, p) < \left(4 + \frac{C_2}{\lg \lg r}\right) \sum_{i=1}^d \gamma_i \Cpolystar(r_i, p) + C_3.
\end{equation}
\end{prop}
\begin{proof}
We first apply Proposition \ref{prop:Cpoly-bound} to construct a $p$-admissible length $N$ such that \eqref{eq:N-bound-2} and \eqref{eq:Cpolystar-bound} both hold; then we apply Proposition \ref{prop:Cad-bound} to construct integers $r_1, \ldots, r_d$ and weights $\gamma_1, \ldots, \gamma_d$ satisfying \eqref{eq:ri-bound} and \eqref{eq:Cad-bound}.
Define $\Phi(x) := 2^{(\log \log x)^8}$; then certainly $r_i < 2^{(\lg \lg 3r)^7} < \Phi(r)$ for large $r$.
The bound \eqref{eq:Cpoly-bound} follows immediately by substituting \eqref{eq:Cad-bound} into \eqref{eq:Cpolystar-bound}.
\end{proof}

Now we may prove our main result for multiplication in $\FF_p[X]$.
The proof is very similar to that of \cite[Thm.~1.1]{HvdHL-ffmul}.
\begin{proof}[Proof of Theorem \ref{thm:poly}]
We have already noted that $\Cpolystar(r, p) = O(K_\ZZ^{\log^* p})$ in the region $r \leq p^2$ (see~\eqref{eq:Cpolystar-small-r}).
To handle the case $r > p^2$, let $z_5$, $C_2$, $C_3$ and $\Phi(x)$ be as in Proposition \ref{prop:Cpoly-recurrence}.
Increasing~$z_5$ if necessary, we may assume that $z_5 > \exp(\exp(1))$ and that $\Phi(x) \leq x-1$ for all $x > z_5$.
For each prime $p$, set $\sigma_p := \max(z_5, p^2)$ and
 \[ L_p := \max(C_3, \max_{2 \leq r \leq \sigma_p} \Cpolystar(r, p)) = O(K_\ZZ^{\log^* p}). \]
Now apply Proposition \ref{prop:master} with $K = 4$, $B = C_2/4$, $\Scal = \{1, 2, \ldots\}$, $\ell = 2$, $\kappa = 1$, $x_0 = x_1 = z_5$, $\sigma = \sigma_p$, $L = L_p$, and $T(r) = \Cpolystar(r, p)$.
The first part of the recurrence for $T(y)$ is satisfied due to the definition of $L_p$, and the second part due to Proposition \ref{prop:Cpoly-recurrence}.
We conclude that $\Cpolystar(r, p) = O(L_p \, 4^{\log^* r - \log^* \sigma_p})$ for $r > p^2$.
Since $\log^* \sigma_p = \log^* p + O(1)$ and $L_p = O(K_\ZZ^{\log^*p})$, we obtain the desired bound $\Cpolystar(r, p) = O(4^{\log^* r - \log^* p} K_\ZZ^{\log^* p})$ for $r > p^2$.
\end{proof}


\section{Faster integer multiplication}
\label{sec:int}

The goal of this section is to prove Theorem \ref{thm:int}.
We will describe a recursive routine \textsc{IntegerMultiply}, that takes as input positive integers $n$ and $t$, and integers $u_1, \ldots, u_t, v \in \ZZ/(2^n - 1)\ZZ$, and computes the products $u_1 v, \ldots, u_t v$.
We denote its running time by $\Cint(t, n)$.
As in Section \ref{sec:poly}, it is convenient to define the normalisation
 \[ \Cintstar(n) := \sup_{t \geq 1} \frac{\Cint(t, n)}{(2t+1) n \lg n}. \]
We certainly have $\MM(n) < \Cint(1, 2n) + O(n)$, so to prove Theorem \ref{thm:int} it is enough to prove that
\begin{equation}
\label{eq:Cintstar-goal}
 \Cintstar(n) = O((4\sqrt 2)^{\log^* n}).
\end{equation}

We begin by revisiting the polynomial multiplication algorithm from Section \ref{sec:poly}.
Recall that to handle a multiplication problem in $\FF_p[X]/(X^r - 1)$ for $r \leq p^2$, we used Kronecker substitution to convert it to an integer multiplication problem of size $O(r \lg p)$ (see \eqref{eq:Cpolystar-small-r}).
This approach is suboptimal because it ignores the cyclic structure of $\FF_p[X]/(X^r - 1)$.

To exploit this structure, we introduce new routines \textsc{RefinedPolynomialMultiply} and \textsc{RefinedAdmissibleMultiply}.
They have exactly the same interface as \textsc{PolynomialMultiply} and \textsc{AdmissibleMultiply}.
Their running times are denoted by $\Crefpoly(t, r, p)$ and $\Crefad(t, N, p)$, with corresponding normalisations $\Crefpolystar(r, p)$ and $\Crefadstar(N, p)$.
The implementation of \textsc{RefinedAdmissibleMultiply} is exactly the same as \textsc{AdmissibleMultiply}, except that calls to \textsc{PolynomialMultiply} are replaced by calls to \textsc{RefinedPolynomialMultiply}.
Similarly, the implementation of \textsc{RefinedPolynomialMultiply} for $r > p^2$ is exactly the same as \textsc{PolynomialMultiply}, except that calls to \textsc{AdmissibleMultiply} are replaced by calls to \textsc{RefinedAdmissibleMultiply}.
Therefore these routines satisfy analogues of Proposition~\ref{prop:Cpoly-bound} and Proposition~\ref{prop:Cad-bound}, with $\Cpolystar$ and $\Cadstar$ replaced by $\Crefpolystar$ and $\Crefadstar$.

Where the new routines differ is in the implementation of \textsc{RefinedPolynomialMultiply} for the case $r \leq p^2$, which is described in the proof of the following result.
The idea is to exploit the cyclic structure by using \textsc{IntegerMultiply} to handle the (cyclic) integer multiplication.
This device saves a constant factor at each recursion level of the main algorithm.

\begin{prop}
\label{prop:Cpoly-bound2}
For any prime $p$, and for any positive integer $r$ satisfying
 \[ (\lg \lg p)^2 < \lg r < (\lg p)^{1/2}, \]
there exists an integer $n$ in the interval
\begin{equation}
\label{eq:n-bound}
 2r \lg p < n < \left(1 + \frac{1}{\lg r}\right) 2 r \lg p
\end{equation}
such that
 \[ \Crefpolystar(r, p) < \left(2 + \frac{O(1)}{\lg r}\right) \Cintstar(n) + O(1). \]
\end{prop}
(Note that this bound does not hold over the whole range $r \leq p^2$; to obtain the constant $2$, we need to restrict to a smaller range of $r$.)
\begin{proof}
We are given as input $U_1, \ldots, U_t, V \in \FF_p[X]/(X^r - 1)$, and we wish to compute the products $U_1 V, \ldots, U_t V$.

We use the following algorithm.
Lift the inputs to polynomials $U'_1, \ldots, U'_t, V' \in \ZZ[X]/(X^r - 1)$, whose coefficients lie in the interval $0 \leq x < p$.
Evaluate these polynomials at $X = 2^b$, where $b := 2 \lg p + \lg r$; that is, pack the coefficients together to obtain integers $u_s := U_s(2^b)$ and $v := V(2^b)$ in $\ZZ/(2^{rb} - 1)\ZZ$.
Call \textsc{IntegerMultiply} with $n := rb$ to compute the cyclic integer products $w_s := u_s v$.
Then we have $w_s = W'_s(2^b)$ where $W'_s := U'_s V' \in \ZZ[X]/(X^r - 1)$.
Observe that the coefficients of $W'_s$ lie in the interval $0 \leq x \leq r(p-1)^2 < rp^2 - 1$, and since $2^b \geq rp^2$, we may unpack $w_s$ to recover the coefficients of $W'_s$ unambiguously.
Finally, by reducing the coefficients of $W'_s$ modulo~$p$, we arrive at the desired products $W_s \in \FF_p[X]/(X^r - 1)$.

As $n = 2r \lg p + r \lg r$, the bound \eqref{eq:n-bound} follows by taking into account the hypothesis that $\lg r < (\lg p)^{1/2}$.
For the complexity we have
 \[ \Crefpoly(t, r, p) < \Cint(t, n) + O(tr \lg p \lg \lg p \lg \lg \lg p), \]
where the last term covers the divisions by $p$ at the end of the algorithm (and also the linear-time packing and unpacking steps).
Dividing by $(2t+1) r \lg p \lg(r \lg p)$ and taking suprema over $t \geq 1$, we obtain
 \[ \Crefpolystar(r, p) < \frac{n \lg n}{r \lg p \lg(r \lg p)} \Cintstar(n) + O\left( \frac{\lg \lg p \lg \lg \lg p}{\lg(r \lg p)}\right). \]
The last term lies in $O(1)$ thanks to the assumption $\lg r > (\lg \lg p)^2$.
Moreover, \eqref{eq:n-bound} implies that $\lg n \leq \lg(r \lg p) + 2$, so we find that
 \[ \frac{n \lg n}{r \lg p \lg(r \lg p)} < 2 \left(1 + \frac{1}{\lg r}\right)\left(1 + \frac{2}{\lg(r \lg p)}\right) < 2 + \frac{O(1)}{\lg r}. \qedhere \]
\end{proof}

Now we describe the implementation of \textsc{IntegerMultiply}.
It chooses one of two algorithms, depending on the size of $n$.
For small $n$, it calls any convenient basecase multiplication algorithm, such as the Sch\"onhage--Strassen algorithm.
For large $n$, it uses the algorithm described in the proof of Proposition \ref{prop:Cint-bound} below.
This algorithm reduces the problem to a collection of instances of \textsc{RefinedAdmissibleMultiply}, one for each prime $p \in \Pset(n)$, where $\Pset(n)$ is defined to be the set consisting of the smallest $\lg n$ primes that exceed $\frac12(\lg n)^2$ (we will see in the proof below that these primes satisfy $\lg p = 2 \lg \lg n + O(1)$).
For example, $\Pset(10^5) = \{149, 151, \ldots, 233\}$ (the first 14 primes after 144.5).

\begin{prop}
\label{prop:Cint-bound}
There is an absolute constant $z_7 > 0$ with the following property.
For all $n > z_7$, there exists an admissible length $N$ in the interval
\begin{equation}
\label{eq:N-bound-3}
 \frac{n}{\lg n \lg \lg n} < N < \left(1 + \frac{3}{\lg \lg n}\right) \frac{n}{\lg n \lg \lg n},
\end{equation}
such that $N$ is $p$-admissible for all $p \in \Pset(n)$, and such that
\begin{equation}
\label{eq:Cintstar-bound}
 \Cintstar(n) < \left(2 + \frac{O(1)}{\lg \lg n}\right) \sum_{p \in \Pset(n)} \frac{1}{\lg n}  \Crefadstar(N, p) + O(1).
\end{equation}
\end{prop}
\resetstep
\begin{proof}
We are given as input $u_1, \ldots, u_t, v \in \ZZ/(2^n - 1)\ZZ$, and we wish to compute the products $u_1 v, \ldots, u_t v$.

\step{choose parameters.}
In this preliminary step we compute a number of parameters that depend only on $n$.

The prime number theorem (see for example \cite[p.~9]{Apo-analytic}) implies that the number of primes between $\frac12(\lg n)^2$ and $(\lg n)^2$ is asymptotically
 \[ \frac{(\lg n)^2}{2 \log((\lg n)^2)} \gg \lg n. \]
Therefore, for large $n$ we certainly have
\begin{equation}
\label{eq:p-bound}
 \tfrac12(\lg n)^2 < p < (\lg n)^2
\end{equation}
for all $p \in \Pset(n)$.
Define $P := \prod_{p \in \Pset(n)} p$; then
 \[ (2 \log_2 \lg n - 1) \lg n \leq \log_2 P \leq (2 \log_2 \lg n) \lg n, \]
so
\begin{equation}
\label{eq:lgP-bound}
 2 \lg n \lg \lg n - 3 \lg n \leq \lg P \leq 2 \lg n \lg \lg n.
\end{equation}
Clearly we may compute $\Pset(n)$ and $P$ within $(\lg n)^{O(1)}$ bit operations.

Let
\begin{equation}
\label{eq:nprime-defn}
 n' := \left \lceil \frac{2n}{\lg P - \lg n - 3} \right \rceil;
\end{equation}
this makes sense for large $n$, as $\lg P \gg \lg n$.
Using Proposition \ref{prop:admissible} (with $p = 2$), construct an admissible length $N$ in the interval
 \[ n' < N < \left(1 + \frac{1}{\lg n'}\right) n'. \]
The invocation of Proposition \ref{prop:admissible} costs $2^{O((\lg \lg n')^2)} = O(n)$ bit operations.
Let us check that \eqref{eq:N-bound-3} holds for this choice of $N$.
In one direction, by \eqref{eq:lgP-bound} we have
\begin{equation}
\label{eq:aux1}
  N > n' \geq \frac{2n}{\lg P - \lg n - 3} \geq \frac{2n}{\lg P} \geq \frac{n}{\lg n \lg \lg n}.
\end{equation}
For the other direction we have
\begin{align*}
  N & <  \left(1 + \frac{1}{\lg n'}\right) \left(\frac{2n}{2 \lg n \lg \lg n - 4 \lg n - 3} + 1\right) \\
    & = \left(1 + \frac{1}{\lg n'}\right) \left(1 + \frac{4 \lg n + 3}{2 \lg n \lg \lg n - 4 \lg n - 3} + \frac{\lg n \lg \lg n}{n}\right) \frac{n}{\lg n \lg \lg n} \\
    & < \left(1 + \frac{1}{\lg n'}\right) \left(1 + \frac{2+o(1)}{\lg \lg n}\right) \frac{n}{\lg n \lg \lg n} \\
    & < \left(1 + \frac{3}{\lg \lg n}\right) \frac{n}{\lg n \lg \lg n}
\end{align*}
for large $n$.

Finally, let us verify that $N$ is $p$-admissible for all $p \in \Pset(n)$ (for large $n$).
First, by \eqref{eq:aux1} and \eqref{eq:p-bound} we have $N > (\lg n)^4 > p^2$.
Also, by \eqref{eq:q-bound} and \eqref{eq:aux1}, every prime divisor of $N = q_0 \cdots q_e$ satisfies $q_i > (\lg N)^3 > (\lg n)^2 > p$; in particular, $p \ndivides N$.

\step{convert to polynomial product modulo $P$.}
In this step we apply the Crandall--Fagin algorithm from Section \ref{sec:crandall-fagin}.
For this, we require that $\lg P > 2\lceil n/N \rceil + \lg N + 1$;
this follows from \eqref{eq:nprime-defn} as
 \[ \lg P \geq \frac{2n}{n'} + \lg n + 3 > \frac{2n}{N} + \lg n + 3 \geq 2 \left\lceil \frac{n}{N} \right\rceil + \lg N + 1. \]
We also require an element $\theta \in \ZZ/P\ZZ$ such that $\theta^N = 2$.
To construct $\theta$, we first compute $a_p := N^{-1} \pmod{p-1}$ for each $p \in \Pset(n)$.
This modular inverse exists because $N$ is a product of primes that are all greater than~$p$, and hence relatively prime to $p-1$.
Then we put $\theta_p := 2^{a_p} \pmod p$, so that $(\theta_p)^N = 2 \pmod p$.
Using the Chinese remainder theorem, we compute $\theta \in \ZZ/P\ZZ$ such that $\theta = \theta_p \pmod p$ for all $p \in \Pset(n)$; then $\theta^N = 2$ as desired.
All of this can be effected in $(\lg n)^{O(1)}$ bit operations.

According to Section \ref{sec:crandall-fagin}, the problem of computing $u_1 v, \ldots, u_t v$ reduces to computing $U_1 V, \ldots, U_t V$ for certain polynomials $U_1, \ldots, U_t, V \in (\ZZ/P\ZZ)[X]/(X^N - 1)$, plus auxiliary operations amounting to
 \[ O(t(N \lg P + N (\lg n)^2 + N \lg P \lg \lg P \lg \lg \lg P)) = O(t n \lg n / \lg \lg n) \]
bit operations.

\step{reduce to products modulo small primes.}
In this step we convert each multiplication problem in $(\ZZ/P\ZZ)[X]/(X^N - 1)$ into a collection of products in $\FF_p[X]/(X^N - 1)$, for $p \in \Pset(n)$.

We start with the isomorphism $\ZZ/P\ZZ \cong \oplus_{p \in \Pset(n)} \FF_p$, which may be computed in either direction in $O(\lg P (\lg \lg P)^2 \lg \lg \lg P)$ bit operations using fast simultaneous modular reduction and fast Chinese remaindering algorithms \cite[\S10.3]{vzGG-compalg3}.
It induces an isomorphism
 \[ (\ZZ/P\ZZ)[X]/(X^N - 1) \cong \bigoplus_{p \in \Pset(n)} \FF_p[X]/(X^N - 1), \]
which may be computed in either direction, for all $s = 1, \ldots, t$, in
 \[ O(tN \lg P (\lg \lg P)^2 \lg \lg \lg P) = O(tn (\lg \lg n)^2 \lg \lg \lg n) \]
bit operations.
Note that the isomorphism $\ZZ/P\ZZ \cong \oplus_p \FF_p$ must be applied to the coefficient of each $X^i$ independently, but the subroutine for multiplying in $\FF_p[X]/(X^N - 1)$ needs sequential access to all of the residues for a single prime~$p$.
The required data rearrangement corresponds to transposing a $tN \times |\Pset(n)|$ array, which costs only $O(tN |\Pset(n)| \lg |\Pset(n)| \max_p \lg p) = O(tn \lg \lg n)$ bit operations.
Finally, for each $p \in \Pset(n)$, the products in $\FF_p[X]/(X^N - 1)$ may be computed by calling \textsc{RefinedAdmissibleMultiply}, since $N$ is a $p$-admissible length.

Combining the contributions from Steps 1, 2 and 3, we obtain
 \[ \Cint(t, n) < \sum_{p \in \Pset(n)} \Crefad(t, N, p) + O(t n \lg n / \lg \lg n). \]
Dividing by $(2t+1) n \lg n$ and taking suprema over $t \geq 1$ yields
 \[ \Cintstar(n) < \sum_{p \in \Pset(n)} \Crefadstar(N, p) \frac{N \lg p \lg(N \lg p)}{n \lg n} + O(1). \]
By \eqref{eq:p-bound} we have $\lg p \leq 2 \lg \lg n$, so \eqref{eq:N-bound-3} implies that
 \[ N \lg p < \left(2 + \frac{O(1)}{\lg \lg n} \right) \frac{n}{\lg n}, \]
and also $\lg(N \lg p) \leq \lg n$, for large $n$.
The bound \eqref{eq:Cintstar-bound} follows immediately.
\end{proof}

We may now glue together the various pieces to obtain a doubly-exponential recurrence for $\Cintstar(n)$.
\begin{prop}
\label{prop:int-bound2}
There are absolute constants $z_9 > z_8 > 0$ and $C_5, C_6 > 0$, and a logarithmically slow function $\Psi : (z_8,\infty) \to \RR$ such that $\Psi(z_9) > z_8$, with the following property.
For any $n > z_9$, there exist positive integers $n_1, \ldots, n_d < \Psi(\Psi(n))$, and weights $\gamma_1, \ldots, \gamma_d > 0$ with $\sum_i \gamma_i = 1$, such that
\begin{equation}
\label{eq:Cint-bound}
 \Cintstar(n) < \left(32 + \frac{C_5}{\lg \lg \lg n}\right) \sum_{i=1}^d \gamma_i \Cintstar(n_i) + C_6.
\end{equation}
\end{prop}
\begin{proof}
\resetstep
\step{top-level call to \textsc{IntegerMultiply}.}
Applying Proposition \ref{prop:Cint-bound}, we obtain an admissible $N$ in the interval
\begin{equation}
\label{eq:N-bound-repeat}
 \frac{n}{\lg n \lg \lg n} < N < \left(1 + \frac{3}{\lg \lg n}\right) \frac{n}{\lg n \lg \lg n},
\end{equation}
such that $N$ is $p$-admissible for all $p \in \Pset(n)$, and such that
\begin{equation}
\label{eq:Cintstar-bound-repeat}
 \Cintstar(n) < \left(2 + \frac{O(1)}{\lg \lg n}\right) \sum_{p \in \Pset(n)} \frac{1}{\lg n} \Crefadstar(N, p) + O(1).
\end{equation}

In what follows, we frequently use the estimates $\lg N = \lg n + O(\lg \lg n)$ and $\lg \lg N = \lg \lg n + O(1)$, which follow from \eqref{eq:N-bound-repeat}.
Also, from \eqref{eq:p-bound} we have $\lg p = 2 \lg \lg n + O(1)$ for all $p \in \Pset(n)$.

\step{first call to \textsc{RefinedAdmissibleMultiply}.}
In this step we use (the refined analogue of) Proposition \ref{prop:Cad-bound} to estimate the $\Crefadstar(N, p)$ term in \eqref{eq:Cintstar-bound-repeat}, for a fixed $p \in \Pset(n)$.
We obtain integers $r_{p,1}, \ldots, r_{p,d_p}$ such that
\begin{equation}
\label{eq:rpi-bound}
 2^{(\lg \lg N)^6} < r_{p,i} < 2^{(\lg \lg N)^7}
\end{equation}
and weights $\gamma_{p,1}, \ldots, \gamma_{p,d_p} > 0$ with $\sum_i \gamma_{p,i} = 1$, such that
 \[ \Crefadstar(N, p) < \left(2 + \frac{O(1)}{\lg \lg N}\right) \sum_{i=1}^{d_p} \gamma_{p,i} \Crefpolystar(r_{p,i}, p) + O(1). \]
Substituting into \eqref{eq:Cintstar-bound-repeat}, and using $\lg \lg N = \lg \lg n + O(1)$, yields
\begin{equation}
\label{eq:Cintstar-bound2}
 \Cintstar(n) < \left(4 + \frac{O(1)}{\lg \lg n}\right) \sum_{p \in \Pset(n)} \sum_{i=1}^{d_p} \frac{\gamma_{p,i}}{\lg n} \Crefpolystar(r_{p,i}, p) + O(1).
\end{equation}

\step{first call to \textsc{RefinedPolynomialMultiply}.}
In this step we use (the refined analogue of) Proposition \ref{prop:Cpoly-bound} to estimate the $\Crefpolystar(r_{p,i}, p)$ term in \eqref{eq:Cintstar-bound2}, for a fixed $p \in \Pset(n)$ and $i \in \{1, \ldots, d_p\}$.
The precondition $r_{p,i} > \max(z_3, p^2)$ holds for large $n$, because by \eqref{eq:rpi-bound} we have
 \[ \lg(p^2) = 4 \lg \lg n + O(1) < (\lg \lg N)^6 \leq \lg r_{p,i}. \]
Thus there exists a $p$-admissible length $N_{p,i}$ in the interval
\begin{equation}
\label{eq:Npi-bound}
  2 r_{p,i} < N_{p,i} < \left(1 + \frac{1}{\lg \lg n}\right) 2 r_{p,i}
\end{equation}
such that
 \[ \Crefpolystar(r_{p,i}, p) < \left(2 + \frac{O(1)}{\lg \lg n}\right) \Crefadstar(N_{p,i},p) + O(1). \]
Substituting into \eqref{eq:Cintstar-bound2} yields
\begin{equation}
\label{eq:Cintstar-bound3}
 \Cintstar(n) < \left(8 + \frac{O(1)}{\lg \lg n}\right) \sum_{p \in \Pset(n)} \sum_{i=1}^{d_p} \frac{\gamma_{p,i}}{\lg n} \Crefadstar(N_{p,i}, p) + O(1).
\end{equation}

\step{second call to \textsc{RefinedAdmissibleMultiply}.}
In this step we use Proposition~\ref{prop:Cad-bound} again, to estimate the $\Crefadstar(N_{p,i}, p)$ term in \eqref{eq:Cintstar-bound3}, for a fixed $p \in \Pset(n)$ and $i \in \{1, \ldots, d_p\}$.
We obtain integers $r_{p,i,1}, \ldots, r_{p,i,d_{p,i}}$ such that
\begin{equation}
\label{eq:rpij-bound}
 2^{(\lg \lg N_{p,i})^6} < r_{p,i,j} < 2^{(\lg \lg N_{p,i})^7}
\end{equation}
and weights $\gamma_{p,i,1}, \ldots, \gamma_{p,i,d_{p,i}} > 0$ with $\sum_j \gamma_{p,i,j} = 1$, such that
 \[ \Crefadstar(N_{p,i}, p) < \left(2 + \frac{O(1)}{\lg \lg N_{p,i}}\right) \sum_{j=1}^{d_{p,i}} \gamma_{p,i,j} \Crefpolystar(r_{p,i,j}, p) + O(1). \]
We have $\lg \lg N_{p,i} \geq \lg \lg r_{p,i} \geq \lg \lg \lg n$, so substituting into \eqref{eq:Cintstar-bound3} yields
\begin{equation}
\label{eq:Cintstar-bound4}
 \Cintstar(n) < \left(16 + \frac{O(1)}{\lg \lg \lg n}\right) \sum_{p \in \Pset(n)} \sum_{i=1}^{d_p} \sum_{j=1}^{d_{p,i}} \frac{\gamma_{p,i} \gamma_{p,i,j}}{\lg n} \Crefpolystar(r_{p,i,j}, p) + O(1).
\end{equation}

\step{second call to \textsc{RefinedPolynomialMultiply}.}
In this step we use Proposition~\ref{prop:Cpoly-bound2} to estimate the $\Crefpolystar(r_{p,i,j}, p)$ term in \eqref{eq:Cintstar-bound4}, for a fixed $p \in \Pset(n)$, $i \in \{1, \ldots, d_p\}$ and $j \in \{1, \ldots, d_{p,i}\}$.
The precondition
 \[ (\lg \lg p)^2 < \lg r_{p,i,j} < (\lg p)^{1/2} \]
holds for large $n$, as \eqref{eq:rpij-bound}, \eqref{eq:Npi-bound} and \eqref{eq:rpi-bound} imply that
 \[ (6 \lg \lg \lg n + O(1))^6 < \lg r_{p,i,j} < (7 \lg \lg \lg n + O(1))^7, \]
whereas $(\lg \lg p)^2 = (\lg \lg \lg n + O(1))^2$ and $(\lg p)^{1/2} = (2 \lg \lg n + O(1))^{1/2}$.
We thus obtain an integer $n_{p,i,j}$ in the interval
\begin{equation}
\label{eq:npij-bound}
 2r_{p,i,j} \lg p < n_{p,i,j} < \left(1 + \frac{1}{\lg \lg \lg n}\right) 2 r_{p,i,j} \lg p
\end{equation}
such that
 \[ \Crefpolystar(r_{p,i,j}, p) < \left(2 + \frac{O(1)}{\lg \lg \lg n}\right) \Cintstar(n_{p,i,j}) + O(1). \]
Substituting into \eqref{eq:Cintstar-bound4} produces
 \[ \Cintstar(n) < \left(32 + \frac{O(1)}{\lg \lg \lg n}\right) \sum_{p \in \Pset(n)} \sum_{i=1}^{d_p} \sum_{j=1}^{d_{p,i}} \frac{\gamma_{p,i} \gamma_{p,i,j}}{\lg n} \Cintstar(n_{p,i,j}) + O(1). \]
The weights $\gamma_{p,i} \gamma_{p,i,j} / \lg n$ sum to $1$, so after appropriate reindexing we obtain the desired bound \eqref{eq:Cint-bound}.

Finally, for the logarithmically slow function $\Psi(x) := 2^{(\log \log x)^8}$, let us verify that for large $n$, we have $n_{p,i,j} < \Psi(\Psi(n))$ for all $p$, $i$ and $j$. First, since $\lg \lg N_{p,i} \geq \lg \lg \lg n$, we have $\lg p < 3 \lg \lg n \leq 3 \cdot 2^{\lg \lg N_{p,i}}$, and hence, by \eqref{eq:npij-bound} and \eqref{eq:rpij-bound},
 \[ n_{p,i,j} < 3 r_{p,i,j} \lg p < 9 \cdot 2^{(\lg \lg N_{p,i})^7 + \lg \lg N_{p,i}} < \Psi(N_{p,i}). \]
Then by \eqref{eq:Npi-bound} and \eqref{eq:rpi-bound} we have
 \[ N_{p,i} < 3 r_{p,i} < 3 \cdot 2^{(\lg \lg N)^7} \leq 3 \cdot 2^{(\lg \lg n)^7} < \Psi(n). \]
Since $\Psi(x)$ is increasing, we get the desired inequality $n_{p,i,j} < \Psi(\Psi(n))$.
\end{proof}

Now we may prove the main theorem for integer multiplication.
\begin{proof}[Proof of Theorem \ref{thm:int}]
We have already noted that it suffices to establish that $\Cintstar(n) = O((4\sqrt 2)^{\log^* n})$ (see \eqref{eq:Cintstar-goal}).
Let $z_8$, $z_9$, $C_5$, $C_6$ and $\Psi(x)$ be as in Proposition \ref{prop:Cpoly-bound2}.
Increasing $z_8$ if necessary, we may assume that $z_8 > \exp(\exp(\exp(1)))$ and that $\Psi(x) \leq x-1$ for all $x > z_8$.
Applying Proposition \ref{prop:master} with $K = 32$, $B = C_5/32$, $\Scal = \{1, 2, \ldots\}$, $\ell = 3$, $\kappa = 2$, $x_0 = z_8$, $x_1 = \sigma = z_9$, $L = \max(C_6, \max_{1 \leq n \leq z_9} \Cintstar(n))$, and $T(n) = \Cintstar(n)$ leads immediately to the desired bound.
\end{proof}


\section*{Acknowledgments}

The authors thank Gr\'egoire Lecerf for his comments on a draft of this paper.
The first author was supported by the Australian Research Council (DP150101689 and FT160100219).

\bibliographystyle{amsplain}
\bibliography{cyclomult}

\end{document}